\documentclass[
  pra,
  onecolumn,
  amsmath,amssymb,
  aps,
  superscriptaddress
]{revtex4-2}

\usepackage{xcolor}
\usepackage[left=2cm,right=2cm,top=20mm,bottom=20mm]{geometry}
\usepackage{graphicx}
\usepackage{enumitem}
\usepackage{hyperref}
\usepackage{amsthm}

\theoremstyle{plain}
\newtheorem{prop}{Proposition}

\begin{document}
\title{Multi-Yoked Surface Codes}

\author{Yuga Hirai}
\email{yuga917@keio.jp}
\affiliation{Department of Applied Physics and Physico-Informatics, Keio University, Kanagawa, Japan}
\affiliation{Center for Quantum Computing, RIKEN, Saitama, Japan}

\author{Yasunari Suzuki}
\email{yasunari.suzuki@riken.jp}
\affiliation{Center for Quantum Computing, RIKEN, Saitama, Japan}

\date{\today}

\begin{abstract}
  The surface code is the leading candidate for fault-tolerant quantum computation owing to its planar nearest-neighbor connectivity, high threshold, and efficient decoding, but its low encoding rate demands on the order of a thousand physical qubits per logical qubit in the teraquop regime. Yoked surface codes improve the density by a factor of two or three by concatenating a high-density parity check code on top of the surface code while retaining planar nearest-neighbor connectivity. In this work, we generalize this idea to multiple levels of concatenation and introduce \textit{multi-yoked surface codes}. Our construction rests on two new ingredients. First, we propose a family of outer codes called \textit{chain codes}, built by chaining copies of a \textit{seed code} that detects the block errors of the code one level below; the resulting codes use all logical qubits of the lower level and achieve a high encoding rate within a few concatenation levels. Second, we design a patch allocation in which movable workspaces perform the yoking checks by lattice surgery with a hierarchy of yoking intervals, so that lower-level checks are performed frequently while higher-level checks are performed less often, which avoids the long yoking intervals that a single large-distance outer code would require. As a benchmark, we construct a three-yoked surface code with a $[[256,166,8]]$ outer code and evaluate it numerically under the SI1000 noise model using soft information from the surface code decoder. The logical error rate is suppressed exponentially in the inner code distance, and in the teraquop regime the three-yoked surface code is about 4 to 4.5$\times$ denser than the unyoked surface code.
\end{abstract}

\maketitle

\section{Introduction}
Quantum computers require error correction during computation as physical qubits are highly susceptible to bit-flip and phase-flip errors. 
The surface code is considered one of the most promising quantum error-correcting codes~\cite{kitaev2003fault,dennis2002topological,fowler2012surface,litinski2019game} because it has several preferable properties: it can be implemented with only nearest-neighbor planar connectivity on a square lattice~\cite{google2023suppressing,google2025quantum}, it has a high threshold of $\sim\! 0.01$~\cite{fowler2009high,wang2009threshold,fowler2012proof}, and its simple structure enables efficient decoding with the minimum-weight perfect matching (MWPM) algorithm~\cite{dennis2002topological,higgott2022pymatching,higgott2023improved,fowler2013optimal} and computation of soft information~\cite{dennis2002topological,gidney2025yoked,meister2024efficient}. 
The surface code has already been demonstrated in real devices, and logical error suppression under realistic device noise has also been experimentally confirmed~\cite{google2023suppressing,eickbusch2025demonstration}.
Despite these advantages, the surface code has a serious drawback: a low encoding rate. Executing a quantum algorithm in the teraquop regime requires a logical error rate per surface-code patch per round below $10^{-15}$~\cite{gidney2025factor}. Assuming $10^{-3}$ physical error rates, this means that a single surface-code patch requires on the order of a thousand physical qubits, including measurement qubits. Therefore, proposing high-density memory space with reasonable hardware requirements is essential for realizing fault-tolerant quantum computation.

There have been many studies aiming to improve the density of memory space, such as quantum low-density parity-check codes and concatenated codes. However, almost all of them impose additional hardware requirements, such as non-nearest-neighbor interactions, high-degree connectivity, deeper syndrome-extraction circuits, or more-than-two-dimensional embeddings~\cite{pattison2025hierarchical,bravyi2024high,higgott2024constructions,steffan2025tile,geher2025directional,nixon2026vine}. These requirements might incur penalties in physical error rates, device scalability, and decoding performance. Thus, methods to improve the density without additional requirements have been lacking.

Recently, yoked surface codes~\cite{gidney2025yoked} have successfully improved density by 2 to 3$\times$ without imposing any additional requirements described above. Ref.~\cite{low2026denser} also improves the density by using a sophisticated arrangement of twist defects. A key idea in their improvements is to concatenate a single high-density parity-check code, such as an $[[n,n-2,2]]$ iceberg code, over the surface-code qubits. In traditional constructions of general concatenated codes, many ancillary code blocks must be allocated for fault-tolerant syndrome extraction, and this overhead undermines the advantage of concatenated codes in meaningful regimes. The yoked surface code eliminates this overhead by performing syndrome extraction of the concatenated code with a carefully designed lattice-surgery sequence with a small workspace.

In this work, we propose \textit{multi-yoked surface codes}, which achieve higher density by concatenating multiple high-density parity-check codes over the surface codes. Specifically, we find that a three-level concatenation, i.e., the yoked-yoked-yoked surface codes illustrated in Fig.~\ref{fig:yoked_yoked_yoked_surface_code}, yields codes 4 to 4.5$\times$ denser than conventional surface codes at a logical error rate below $10^{-15}$.
It is well known that multi-level concatenation significantly improves the encoding rate~\cite{yamasaki2024time,yoshida2025concatenate}, but designing a high-density implementation is nontrivial for the following two reasons. 
First, the multiple concatenations of high-density parity-check codes introduce complicated correlated errors in lower-level code blocks. Therefore, we need to provide a series of high-density parity-check codes that can detect problematic correlated errors in inner codes. Second, we need additional workspace for the syndrome extraction of outer codes. This overhead might cancel the density improvement by multi-level concatenation. Thus, we need to design efficient workspace allocations and syndrome-extraction schedules for outer codes. This work addresses both problems.

\begin{figure}[b]
  \centering
  \includegraphics[width=1\textwidth]{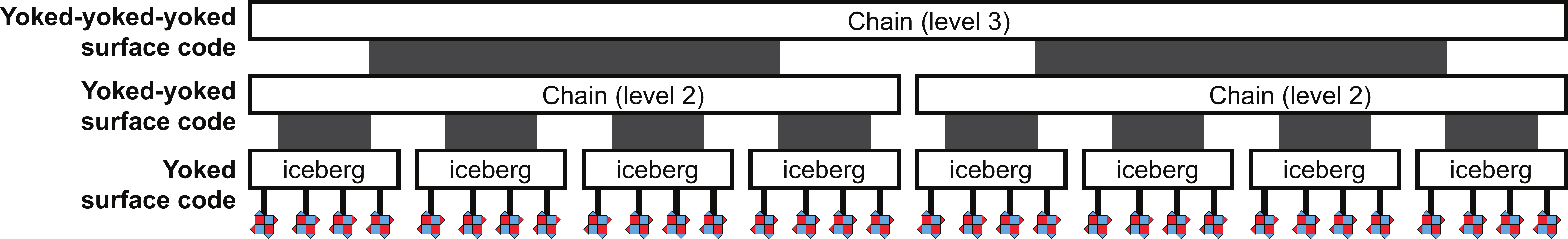}
  \caption{Concept of multi-yoked surface codes.}
  \label{fig:yoked_yoked_yoked_surface_code}
\end{figure}

We resolve the first problem by proposing \textit{chain codes}, which generalize the mechanism of the $C_4/C_6$ code~\cite{knill1996concatenated}. Here, the $C_4$ code is a $[[4,2,2]]$ iceberg code, and the $C_6$ code is a $[[6,2,2]]$ non-iceberg code. Although two logical qubits encoded by the $C_4$ or $C_6$ code might suffer from correlated errors, the $C_6$ code can be concatenated over three $C_4$ or $C_6$ code blocks because the $C_6$ code is capable of detecting these correlated errors. 
In the multi-yoked surface codes, we use surface-code qubits instead of physical qubits, and use general $[[n,n-2,2]]$ iceberg codes instead of $C_4$ codes. To detect correlated errors in the iceberg codes and their underlying inner codes, we propose a framework named \textit{chain codes}, which can detect all correlated errors in each lower-level code block as the $C_6$ code does.
When we want to define chain codes for multiple blocks of a given inner code, we first find a small \textit{seed code} that can detect all problematic errors in one code block. Then, we construct the chain code by chaining copies of the seed code as shown in Fig.~\ref{fig:chain_codes}. We consider up to three-level concatenation for a benchmark in this work; this yields an outer code $[[128n_2,104n_2-42,8]]$, whose encoding rate $k/n$ is $\sim0.65$ for $n_2 = 2$ and $\sim0.80$ for $n_2 = 20$.

For the latter problem, we present an efficient memory-space implementation of the multi-yoked surface code. We separate the whole layout into several modules, allocate workspaces for each high-density parity-check code in the outer code, and move them to perform syndrome extractions. An example layout is shown in Fig.\,\ref{fig:patch_alloc_employed}.
This layout not only enables parity checks for outer codes but also allows fine-tuning of the trade-off between space overhead and the \textit{yoking interval}, that is, the time between two consecutive yoking operations. 
In the original yoked surface code, a space-efficient implementation with a single hallway shown in Fig.~\ref{fig:yoking_config} was adopted for yoking to minimize the space overhead. When the outer code distance becomes larger, however, this design easily lengthens the yoking interval, which in turn requires larger surface codes to tolerate idling errors. Thus, it may achieve high density only in a regime of logical error rates far lower than those required in practice. In our layout, we adopt both time-efficient and space-efficient yoking designs to balance the frequencies of the checks of each level and the number of workspaces. This allows us to use the qubits spent on excess suppression to make the code denser. Note that hook errors in physical syndrome-extraction circuits for this procedure can be systematically eliminated without losing the fault distance by using the technique from our previous work~\cite{hirai2026no}.

We numerically simulated the performance of the multi-yoked surface code with a gap simulation similar to Ref.~\cite{gidney2025yoked}, and optimized parameters for target logical error rates. We find that, with an inner code distance of $d_\text{inner} = 10$, it is about 4 to 4.5$\times$ denser than conventional surface codes under the SI1000 noise model, which corresponds to a 1.3 to 1.5$\times$ density improvement over the yoked surface code. We note that the simulated three-yoked surface code is only a prototype; a more elaborate design, or concatenation beyond three levels, may yield an even higher density.

Finally, we note a drawback of our multi-yoked surface codes compared to the original yoked surface code, although we believe it can be resolved. In the yoked surface codes, both inner and outer codes are matchable, and they can be simultaneously decoded with fast MWPM decoders. In our case, while level-0 and level-1 codes (i.e., surface codes and iceberg codes) are matchable, the outer codes (i.e., chain codes) are not. Therefore, we need to decode the inner parts with MWPM decoders, and then decode outer codes with other methods. 
We believe, however, that this does not pose a serious problem in the whole architecture. This is because the outer code is small and its syndrome-extraction frequency is much lower than that of the inner codes. Thus, we believe its real-time decoding would be feasible, for example, with beam-search-based decoders~\cite{beni2025tesseract}. Architecting such a two-stage decoding in realistic systems goes beyond our scope and is left as future work. 

The rest of this paper is organized as follows. In Section~\ref{sec:preliminary}, we review the basics of surface codes and concatenated codes. In Section~\ref{sec:existing_work}, we introduce the previous work on yoked surface codes and explain what is new in this work. In Section~\ref{sec:proposed_method}, we present the construction of multi-yoked surface codes, including chain codes and the patch allocation. In Section~\ref{sec:numerical_results}, we show the simulation results of multi-yoked surface codes. Finally, in Section~\ref{sec:discussion}, we conclude this paper and discuss open questions.

\section{Preliminary}
\label{sec:preliminary}
Multi-yoked surface codes are based on a concatenated code whose level-0 code is the surface code. In this section, we review the basics of surface codes and concatenated codes, both of which are needed in what follows.

\subsection{Surface Codes}
The surface code~\cite{kitaev2003fault,dennis2002topological,fowler2012surface,litinski2019game} is a stabilizer code defined on a square lattice, where the data qubits sit on the vertices, and both the Z and the X stabilizers sit on the plaquettes in the checkerboard pattern shown on the left in Fig.~\ref{fig:surface_code}. Each plaquette has its own measurement qubit, which extracts the parity of the neighboring data qubits. We call a boundary terminated by X (Z) stabilizers an \textit{X boundary} (\textit{Z boundary}), and an X (Z) logical operator is defined as a chain of X (Z) operators on the data qubits connecting the two X (Z) boundaries. A single patch therefore encodes one logical qubit. The surface code has several attractive properties. First, it requires only planar, nearest-neighbor connectivity on a square lattice, which relaxes the demands on the device. Second, it has a high threshold of $\sim\!0.01$. Third, it can be decoded relatively easily by the minimum-weight perfect matching~(MWPM) algorithm~\cite{higgott2022pymatching}, since every error in the decoding graph can be represented by an edge; a hyperedge corresponding to a Y error is decomposed into the edges corresponding to X and Z errors without significantly degrading the decoding accuracy. We call stabilizer codes with this property \textit{matchable}.

\begin{figure}[h]
  \centering
  \includegraphics[width=0.30\textwidth]{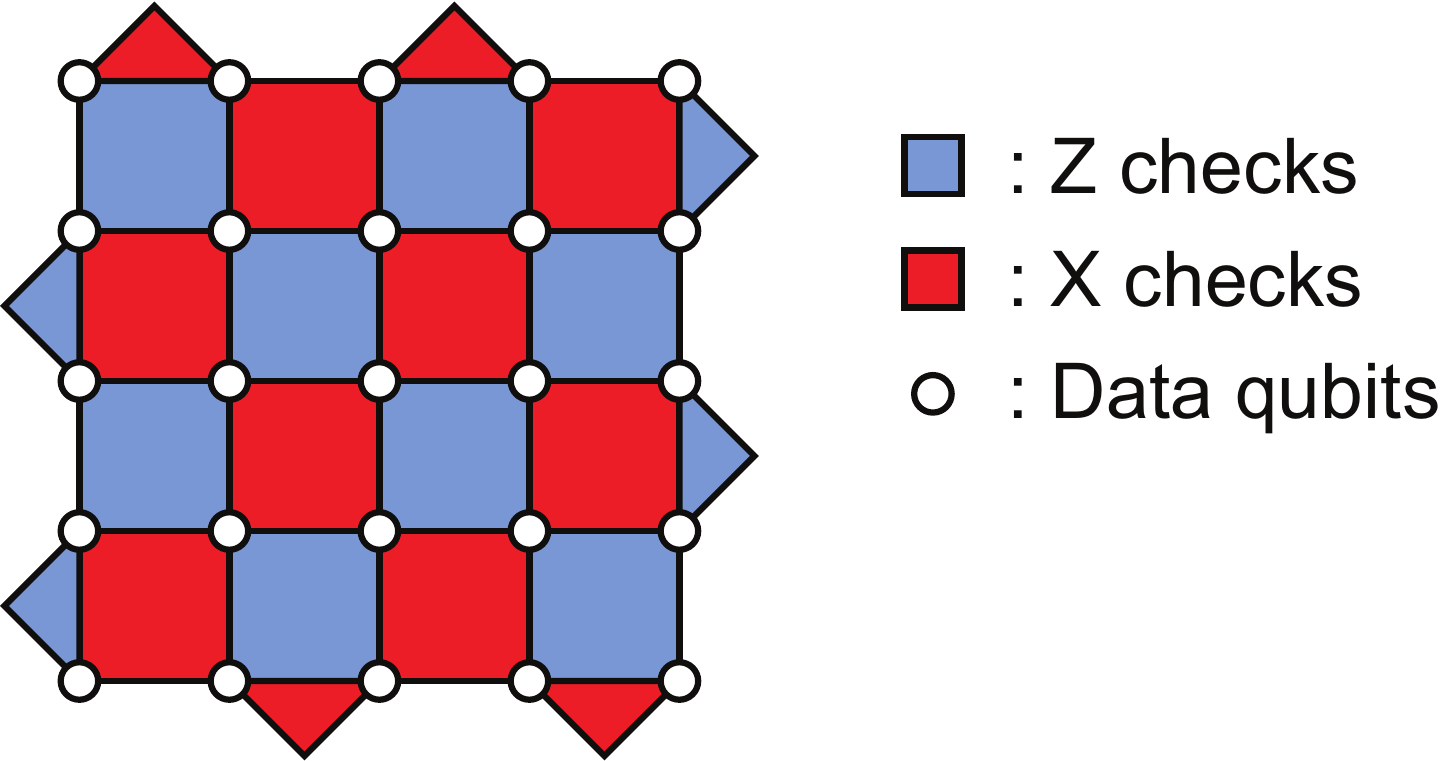}
  \hspace{20pt}
  \includegraphics[width=0.35\textwidth]{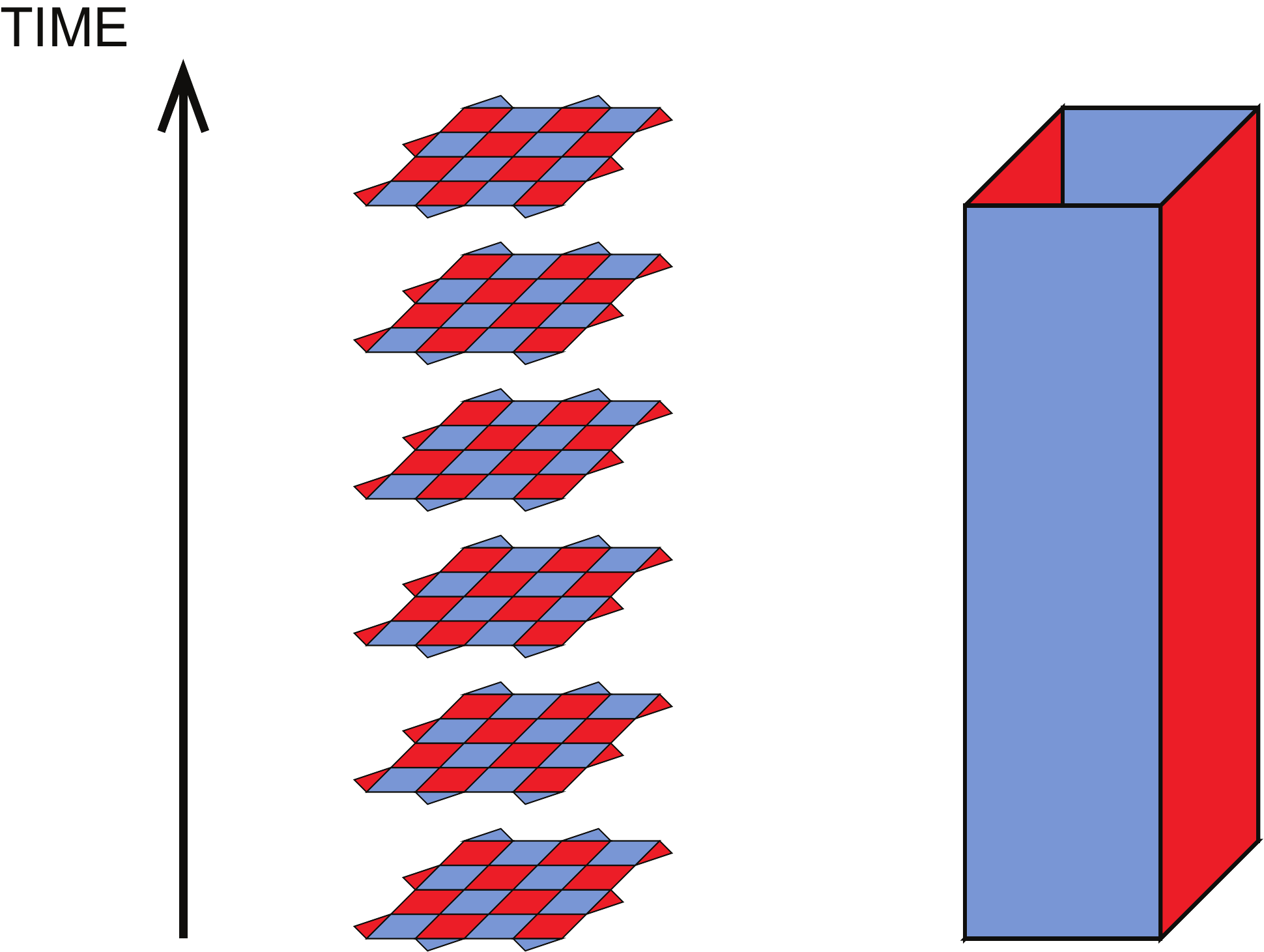}
  \caption{From left to right: the surface code and the defect diagram.}
  \label{fig:surface_code}
\end{figure}

During the computation, the Z and X check operators are repeatedly measured using the ancillary qubits placed at the center of the check operators. If we assume a circuit-level noise model, i.e., a model in which Pauli errors are inserted for each physical gate and idling operation, Pauli errors during parity measurements might propagate to multiple data qubits; such errors are called hook errors. With a carefully designed order of CNOT gates, these hook errors do not reduce the fault distance~\cite{hirai2026no,gidney2025alternating,kishony2026surface}. This is another advantage of surface codes, as the concatenated codes discussed later require large overheads for fault-tolerant parity checks.

Since the surface code has a simple structure, it can be expressed by more abstract diagrams. One of them is the \textit{defect diagram} shown in the right panel of Fig.~\ref{fig:surface_code}. The red (blue) walls represent the X (Z) boundaries, and the empty space enclosed by the walls is the bulk of the surface code, where the stabilizers are defined. The logical operators in this diagram are defined in the same way as before: they are chains connecting two disjoint walls of the same Pauli type. In short, the diagram represents the spacetime trajectory of a surface code patch, and it can describe various operations. When we are restricted to nearest-neighbor interactions, logical operations are performed by a technique called lattice surgery~\cite{horsman2012surface,fowler2018low}. Lattice surgery can be executed by merging two surface-code patches into one and then splitting it back into two, as shown in Fig.~\ref{fig:lattice_surgery}. In this figure, the two patches are merged along their X boundaries, so the operation realizes a joint ZZ measurement between them. Patch rotation and movement can also be regarded as forms of lattice surgery. As we will see later, all the stabilizer measurements of outer codes are realized by this lattice surgery operation on the surface code.

\begin{figure}[h]
  \centering
  \includegraphics[width=0.30\textwidth]{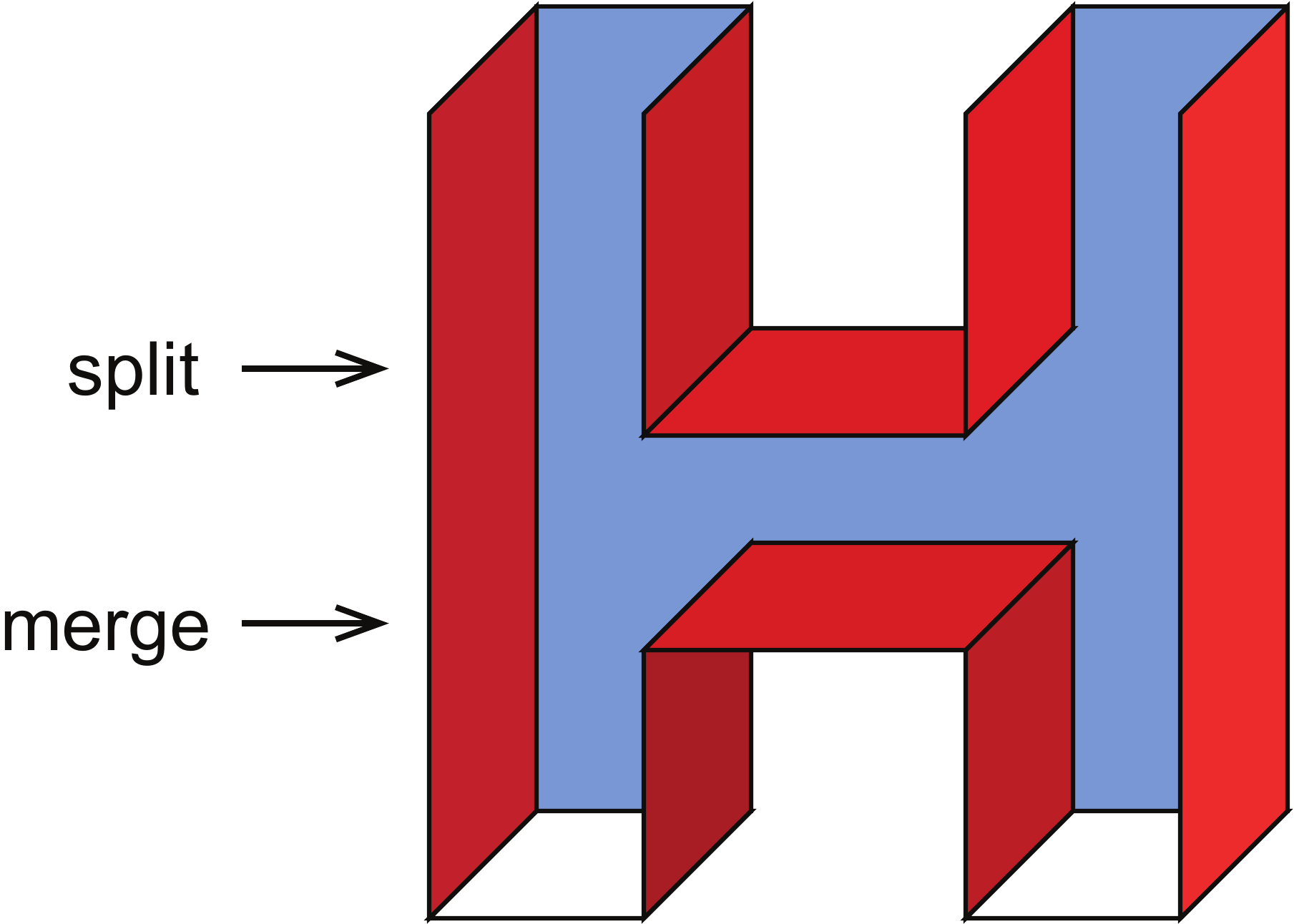}
  \caption{Lattice surgery. The front Z-type wall is made transparent for clarity. Time goes upward.}
  \label{fig:lattice_surgery}
\end{figure}

When we decode errors in the surface code, we can calculate, as a byproduct, soft information, which quantifies how confident the decoder is in its estimate for given syndrome values. One such piece of information is called the \textit{complementary gap}~\cite{gidney2025yoked,bombin2024fault,hutter2014efficient}, which is defined as the log-likelihood ratio between the minimum-weight matching and the minimum-weight matching in the complementary logical class. In the multi-yoked surface code, such a value is assigned to each edge of the outer-code decoding graph. Note that the computational complexity of evaluating the complementary gap grows exponentially as the number of logical qubits increases. Recently, it has been shown that this overhead can be avoided by using a different soft-information measure that is similar to the complementary gap~\cite{meister2024efficient}.

\subsection{Concatenated Codes and $\boldsymbol{C_4/C_6}$ Codes}\label{sec:concatenated_codes}
Concatenated codes are codes in which qubits are recursively encoded into quantum error-correcting codes~\cite{aharonov2008fault,knill1996concatenated,knill2005quantum,yamasaki2024time,yoshida2025concatenate,goto2024high}. Here, the concatenation level is counted starting from the physical qubits. 
One of the best-known examples of concatenated codes is the $C_4/C_6$ code introduced by E.~Knill~\cite{knill1996concatenated,knill2005quantum}. This code consists of two types of codes. One is the $[[4,2,2]]$ $C_4$ code, and the other is the $[[6,2,2]]$ $C_6$ code. The $C_4/C_6$ code first encodes physical qubits with the $C_4$ codes, which are recursively encoded with $C_6$ codes. Fig.~\ref{fig:c4c6} shows the code construction.

In general, encoding more than one logical qubit per block may cause correlated error events at the higher level, and the code distance does not necessarily scale as $d_{\mathrm{lower}} d_{\mathrm{higher}}$, where $d_{\mathrm{lower}}$ and $d_{\mathrm{higher}}$ are the code distances of the lower- and higher-level codes. We call such correlated errors \textit{block errors}. 
A simple solution to this problem is to distribute the $ i$-th-level logical qubits across different code blocks at level-$(i+1)$. However, this increases the minimum number of physical and logical qubits required to achieve a target logical error rate and is not suitable for practical implementation.
The $C_4/C_6$ code resolves this problem by letting the $C_6$ code detect such block errors on a pair of logical qubits, and its code distance does scale as $d_\text{lower}d_\text{higher}$. The code distance of the $C_4/C_6$ code therefore scales as $2^N$, where $N$ is the number of concatenation levels. 

\begin{figure}[h]
  \centering
  \includegraphics[width=0.40\textwidth]{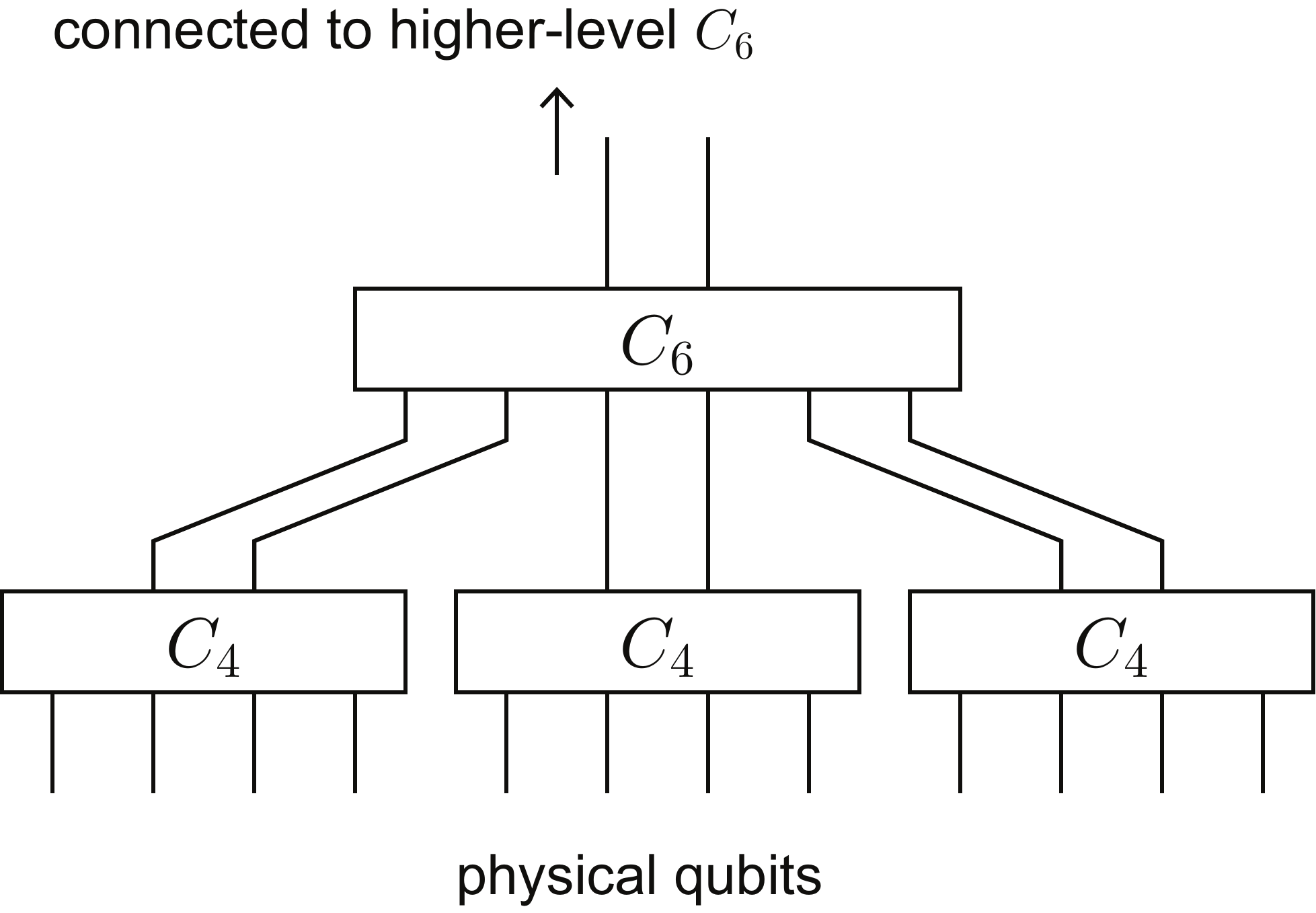}
  \caption{A part of the $C_4/C_6$ code. Each box represents a code block. Reading from bottom to top, the lines entering a box from below are the qubits that constitute the block, and the lines leaving it from above are its logical qubits.}
  \label{fig:c4c6}
\end{figure}

Unlike surface codes, a naive implementation of parity measurements of concatenated codes is not fault-tolerant due to hook errors, and we need to introduce techniques to fault-tolerantly measure parities. One well-known method is Knill's error correction, which implements parity measurements using the Knill gadget~\cite{knill1996concatenated,knill2005quantum}. In this gadget, the parity measurements are performed by state teleportation, which consumes logical Bell pairs at each level of concatenation. Therefore, we have to prepare verified Bell pairs in advance. Because this generation uses error detection and post-selection, it introduces significant spacetime overhead. Other approaches, such as Shor's and Steane's error correction, also introduce similar overheads.

Concatenated codes with higher encoding rates than the $C_4/C_6$ code have been extensively explored. For example, Refs.~\cite{yoshida2025concatenate,yamasaki2024time} achieved a constant space overhead in the asymptotic limit of many concatenation levels. However, the overheads from syndrome-extraction gadgets might significantly degrade the density in the practical regime. Therefore, achieving high-density memory-space on a square-grid qubit array is still challenging. 
Recently, Ref.\,\cite{wills2026concatenating} has proposed concatenating Reed–Solomon codes on top of bivariate bicycle codes, each encoding 12 logical qubits, by treating each lower-level block as a qudit. Although this enables simple concatenation that is tolerant to block errors, this approach requires the outer code to detect all qudit errors. As we will explain later in Sec.\,\ref{sec:chain-code},  our proposal, the chain code, detects only problematic errors among them.

\section{Existing Work: Yoked Surface Codes}
\label{sec:existing_work}
Multi-yoked surface codes are heavily based on the idea of yoked surface codes by C.~Gidney \textit{et al.}~\cite{gidney2025yoked}. In this section, we introduce their concept and explain our strategy to further improve the density.

\subsection{Yoked Surface Codes}
There are many ideas to improve the density of the normal surface code. However, almost all ideas impose strong hardware requirements, such as non-nearest-neighbor interactions, high-degree connectivity, and more-than-two-dimensional embeddings. 
Yoked surface codes need none of them; they require only nearest-neighbor planar connectivity on a square lattice. This connectivity has already been realized in real devices~\cite{google2023suppressing,eickbusch2025demonstration}. 

The strategy of the yoked surface codes for densifying the surface codes is to concatenate a high-density code on top of them. We refer to the surface code as the \textit{inner code} and to the concatenated code as the \textit{outer code}. 
Yoked surface codes employ iceberg codes or quantum parity check codes as the outer code for $d_{\rm outer}=2$ and $d_{\rm outer}=4$, respectively. The left panel of Fig.~\ref{fig:yoked_surface_codes} shows check operators of the $[[6,4,2]]$ iceberg code and logical Pauli operators for four logical qubits. The parity measurements of the outer code in the yoked surface codes are performed by lattice surgery on the surface codes, as shown in the middle panel of Fig.~\ref{fig:yoked_surface_codes}. The probability of causing correlated errors across multiple surface codes, which are analogous to hook errors in the conventional approaches, is sufficiently suppressed by widening the lattice surgery path. The widened paths are extended along the time direction rather than the spatial directions, so as to minimize the workspace needed for the parity checks, as shown in the right panel of the figure. The lattice surgeries for the different checks are arranged sequentially along the time direction.
This construction enables more space-efficient concatenation of high-density parity-check codes than syndrome-extraction gadgets in conventional concatenated codes.
With this syndrome extraction, both inner and outer codes become matchable. Therefore, they can be decoded simultaneously using MWPM decoders by assigning the detectors of the outer code to the corresponding boundary detectors of the inner code.

\begin{figure}[h]
  \centering
  \includegraphics[width=0.80\textwidth]{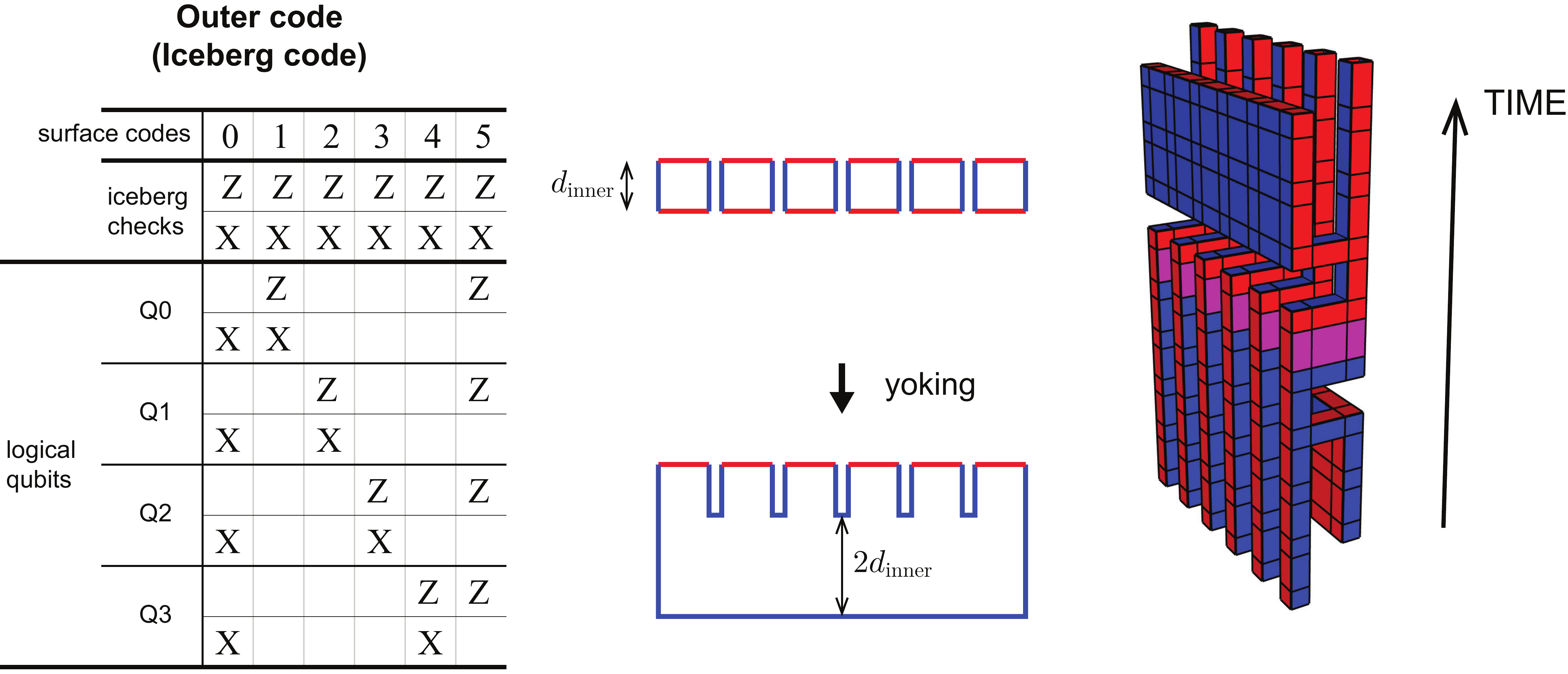}
  \caption{From left to right: outer code (iceberg code), the surface-code lattice surgery for the outer Z check, and the defect diagram representing the Z- and X-type checks with a patch rotation.}
  \label{fig:yoked_surface_codes}
\end{figure}

\subsection{Open Questions We Address}
To further improve the density of the yoked surface code, we should increase the distance of outer codes with high encoding rates and reduce the distance of inner surface codes with low encoding rates. There are two ways to increase the outer code distance: one is to use a single code with large distance, and the other is to use a multi-level concatenated code. 

Ref.\,\cite{gidney2025yoked} already pursued the former approach by replacing $d_{\rm outer}=2$ iceberg codes with a $d_{\rm outer}=4$ quantum parity check code. However, they did not go beyond distance four because increasing the outer code distance does not work straightforwardly for the following reasons.
If we increase the code distance of the outer code, the number of outer checks increases. This makes the interval between two subsequent parity measurements of the outer codes longer. This will eventually increase the idling errors of surface-code qubits and degrade the overall performance. Although this can be mitigated by increasing the inner code distance, this contradicts the original motivation to decrease the inner code distance. Also, increasing the inner code distance makes the available total code distance $d_{\rm inner} d_{\rm outer}$ highly discretized, and typically leads to excessive logical error suppression below what is actually required. The qubits spent on this excess suppression could instead be used to make the code denser. 

Therefore, this paper explores the other direction: increasing outer-code distance by further concatenating $d=2$ high-density parity check codes over the yoked surface codes. This direction is promising because logical error rates of qubits encoded with the level-0 (surface) and level-1 (iceberg) codes are suppressed, and the required interval of the additional outer-code checks can be much longer than that of the iceberg codes. Thus, further concatenations are expected to have little effect on the lower-level check intervals. Our research question is therefore ``How much can we improve the density of the yoked surface code by increasing the concatenation level?''

To address this question, we need to resolve two difficulties. First, if we use a concatenated code as the outer code, its upper-level codes must manage block errors in lower-level codes. Although the $C_6$ code satisfies this property, it is an open question whether we can satisfy this criterion for a lower-level code encoding more than two logical qubits. 
Second, even if a series of high-density parity check codes exists, it remains open whether they can be implemented on a square-grid lattice with modest workspace overhead. 
We address these two open questions by proposing chain codes and memory-space layouts tailored for multi-level concatenated codes.

\section{Multi-yoked Surface Codes}
\label{sec:proposed_method}
We now introduce multi-yoked surface codes, which concatenate additional $d=2$ codes on top of the yoked surface codes. The multi-yoked surface code is built from two new components: a family of codes called chain codes, and a memory-space layout of the surface-code patches. We consider up to 3-level concatenation for multi-yoked surface codes. The code thus consists of level-0 surface codes, level-1 iceberg codes, level-2 chain codes, and level-3 chain codes. Below, we introduce the framework for chain codes and show a construction of level-2 and level-3 chain codes. Then, we show how to allocate surface-code patches to achieve high density and overall memory-space construction.

\subsection{Chain Codes} \label{sec:chain-code}
In Subsection~\ref{sec:concatenated_codes}, we introduced the $C_4/C_6$ code, in which a $C_6$ block detects the block errors of the three $C_4$ blocks one level below. In this subsection, we generalize this concept by constructing a code that can detect all single block errors of the iceberg codes one level below, as shown in Fig.~\ref{fig:generalized_c4c6}. 
We concentrate on distance-two CSS codes concatenated on top of the iceberg code. Note that hook errors during parity checks are discussed in Sec.\,\ref{sec:patch_allocation}, so we ignore the hook errors in this subsection. Also, we assume that Z errors and X errors are estimated independently, and that X and Z errors are detected by Z and X checks, respectively. Therefore, without loss of generality, we focus on the mechanism of detecting X errors with Z checks.

\begin{figure}[h]
  \centering
  \includegraphics[width=0.75\textwidth]{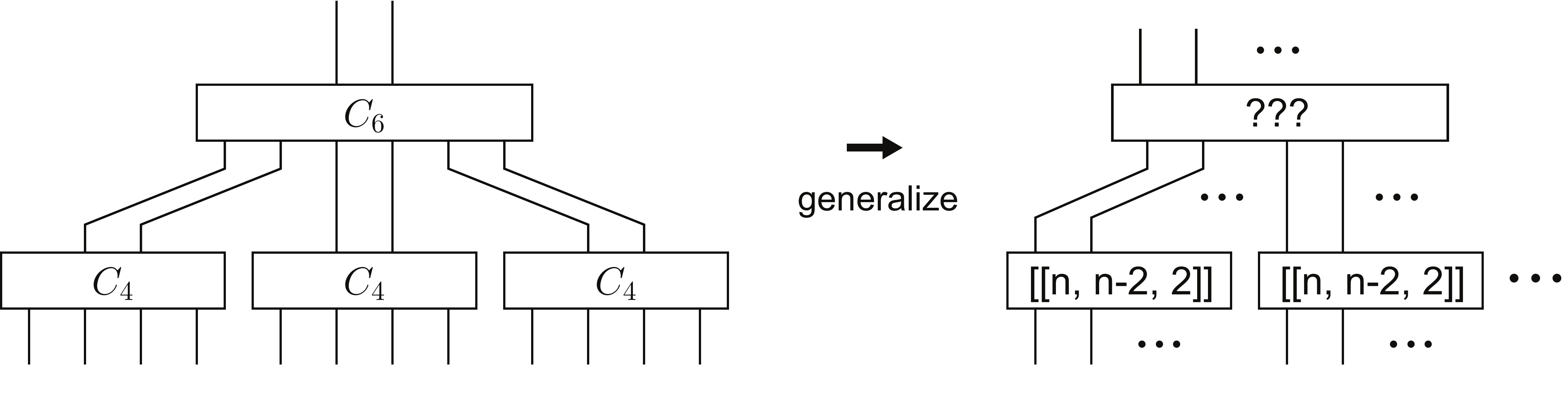}
  \caption{From left to right: A part of the $C_4/C_6$ code and the generalized version of it.}
  \label{fig:generalized_c4c6}
\end{figure}

Suppose level-1 to level-$(i-1)$ concatenated codes constitute an $[[n,k,d]]$ code, and we want to construct a level-$i$ code with distance two that is concatenated over $r$ copies of them. We call every logical X error that can be induced by a weight-$(2d-1)$ undetectable X error on level-0 (i.e., surface-code qubits in the multi-yoked surface code) qubits a \textit{problematic error} and denote their set by $\mathcal{P}_{\rm prob, X} \subset \{I,X\}^{\otimes k}$. 
Next, we list an independent set of Pauli Z operators that can detect all the problematic errors. We call this set of Pauli operators a \textit{seed code} $\mathcal{P}_{\rm seed, Z} \subset \{I,Z\}^{\otimes k}$, which satisfies $\forall P \in \mathcal{P}_{\rm prob, X},\ \exists Q \in \mathcal{P}_{\rm seed, Z},\ {\rm s.t.}\ QP=-PQ$.
Then, the stabilizer generator set of a chain code is constructed as $\{P^{\otimes r} | P \in \mathcal{P}_{\rm seed, Z}\} \cup \{P^{\otimes r} | P \in \mathcal{P}_{\rm seed, X}\}$. Note that elements of X and Z seed codes $\mathcal{P}_{\rm seed, X}$ and $\mathcal{P}_{\rm seed, Z}$ do not necessarily commute with each other. Even in this case, the chain codes satisfy the condition of stabilizer codes if we choose $r$ as an even number. 
Thus, the chain code encodes $rk-|\mathcal{P}_{\rm seed, X}|-|\mathcal{P}_{\rm seed, Z}|$ logical qubits with $rk$ lower-level qubits, and it is capable of detecting all block errors at the lower-level codes. Its encoding rate converges to one in the large-$r$ limit.

A non-trivial part of the above procedure is finding a seed code. A simple method to find a seed code is to choose $\mathcal{P}_{\rm seed, X} = \{X_0, ..., X_{k-1}\},\ \mathcal{P}_{\rm seed, Z} = \{Z_0, ..., Z_{k-1}\}$, and an even $r$. This always detects all problematic errors but results in a low encoding rate of the chain code. Thus, we should minimize the number of members in the seed codes. Although this minimization might be computationally expensive, we can find the minimum set by brute-force or heuristic enumeration as explained below.

\subsection{Finding a Level-2 Chain Code on the Level-1 Iceberg Code}
\label{sec:level2-code}
Using the chain-code framework, we find a high-rate chain code for iceberg codes. The set of problematic errors depends on how we define the logical operators. For the iceberg codes, we define $n-2$ logical Pauli X/Z operators as $X_{L,i} = X_0 X_{i+1}$ and $Z_{L,i} = Z_{i+1} Z_{k+1}$. See the left panel of Fig.\,\ref{fig:yoked_surface_codes} for an example with $[[6,4,2]]$. Then, we enumerate all the problematic errors in iceberg codes. Let $1 \leq i, j \leq k$ be two different integers. To enumerate problematic errors, we need to consider Pauli errors on surface-code qubits up to weight three, but the weight-1 and weight-3 errors are detected by iceberg codes. Therefore, we need only consider weight-2 errors for enumerating problematic errors as undetectable candidates. Any problematic error caused by weight-2 surface-code errors can be classified into the following four patterns:
\begin{enumerate}[label=(\roman*), align=left,]
  \item Single error: $X_0 X_{i}$ corresponds to $X_{L, i-1}$,
  \item Two errors: $X_i X_j$ corresponds to $X_{L, i-1} X_{L,j-1}$,
  \item $k-1$ errors: $X_i X_{k+1}$ corresponds to $\prod_{l\neq i-1, 0\leq l\leq k-1} X_{L,l}$ up to stabilizer,
  \item $k$ errors: $X_0 X_{k+1}$ corresponds to $\prod_{0\leq l\leq k-1}X_{L,l}$ up to stabilizer.
\end{enumerate}
Thus, their union constitutes $\mathcal{P}_{\rm prob, X}$ for the iceberg code. Similarly, we can enumerate all the possible problematic Z errors.

Next, we find a small seed code that detects all problematic X and Z errors. Regarding the number of checks of a seed code, the following proposition holds:
\begin{prop}\label{prop:minimum_checks}
Let $\mathcal{P}_{\rm prob, X}$ be the set of problematic X errors for an iceberg code $[[n,k=n-2,2]]$ and $c_z$ be the size of the Z-check seed code $\mathcal{P}_{\rm seed, Z}$ for $\mathcal{P}_{\rm prob, X}$. Then, $c_z$ is lower bounded as $c_z \geq \lceil \log_2{(k+1)}\rceil$. The same holds for X checks for Z errors.
\end{prop}

\begin{proof}
By the above error classification, the Z-check matrix of a seed code must detect all weight-one and weight-two errors on $k$ logical qubits. Detecting all weight-one errors requires that the columns of the Z-type check matrix be nonzero, and detecting all weight-two errors requires that the columns be pairwise distinct. Therefore, we can upper-bound the number of columns as $2^{c_z}-1\geq k$, and hence $c_z\geq \log_2{(k+1)}$. Since $c_z$ is an integer, we obtain $c_z\geq \lceil \log_2{(k+1)}\rceil$.  The same argument applies to X checks.
\end{proof}

We find such minimum sets of checks by hand for several iceberg codes of interest. Here, we show the seed codes for the $[[8,6,2]]$ and the $[[16,14,2]]$ iceberg codes. We are interested in them because the $[[8,6,2]]$ and the $[[16,14,2]]$ iceberg codes are the highest-rate codes that have a seed code with size $c_z=3$ and $c_z=4$, respectively. The seed code for the $[[8,6,2]]$ iceberg code is listed as follows, where each line corresponds to a single check operator.
\begin{align}
  \begin{array}{cccccc}
    Z&I&I&Z&Z&I\\
    I&Z&I&Z&I&Z\\
    I&I&Z&I&Z&Z\\
    I&X&X&I&I&X\\
    X&I&X&I&X&I\\
    X&X&I&X&I&I\\
  \end{array}
\end{align}
Note that, as elements in the seed code commute with each other, we can choose an even or odd integer $r$ in this case. Then, we horizontally copy this outer-check table, as shown in Fig.\,\ref{fig:chain_codes}. We list the logical Pauli operators for the chain codes, which can be used to find problematic errors in the upper-level code. We also similarly find a seed code for $[[16,14,2]]$ iceberg codes. The list of the check operators is as follows.
\begin{align}
  \begin{array}{cccccccccccccc}
    Z&I&Z&Z&I&Z&I&I&Z&I&I&Z&Z&I\\
    I&Z&I&I&Z&I&Z&Z&Z&I&I&Z&Z&I\\
    I&I&Z&I&Z&Z&I&Z&I&Z&I&Z&I&Z\\
    I&I&I&Z&I&Z&Z&Z&I&I&Z&I&Z&Z\\
    X&I&X&X&I&X&I&I&X&I&I&X&X&I\\
    I&X&I&I&X&I&X&X&X&I&I&X&X&I\\
    I&I&X&I&X&X&I&X&I&X&I&X&I&X\\
    I&I&I&X&I&X&X&X&I&I&X&I&X&X
  \end{array}
\end{align}

\subsection{Finding a Level-3 Chain Code on the Level-2 Chain Code}
\label{sec:level3-code}
Next, we construct a level-3 chain code over level-1 iceberg and level-2 chain codes. We expect the combination of level-1 and level-2 codes to tolerate up to three Pauli errors in surface-code qubits. As the size of the level-2 chain code is typically large, a brute-force approach does not work. Instead, we employ the following heuristic procedure to find a seed code for high-rate level-3 chain codes. 

We assume that the level-1 iceberg code and level-2 chain code constitute $k$ logical qubits per block. Our heuristic exploration for the level-3 code runs with a positive input parameter $c$ as follows. First, we choose a classical code $[k, k - c, d]$ with the largest $d$ from a code table~\cite{codetables}. We use this code's parity-check matrix as the seed code for both Z and X errors. Then, we iteratively permute an assignment of $k$ logical qubits in the level-2 codes to the input of the level-3 seed code. In each trial, we compute the distance of the chain code generated from the level-3 seed code. If the distance reaches eight, the procedure stops and outputs the chain code. If not, the procedure aborts. While small $c$ results in a high-rate level-3 code, this procedure aborts with high probability if the chosen $c$ is too small. Therefore, we start with small $c$ and iterate to larger $c$ until we find a solution. 

\begin{figure}[h]
  \centering
  \includegraphics[width=0.75\textwidth]{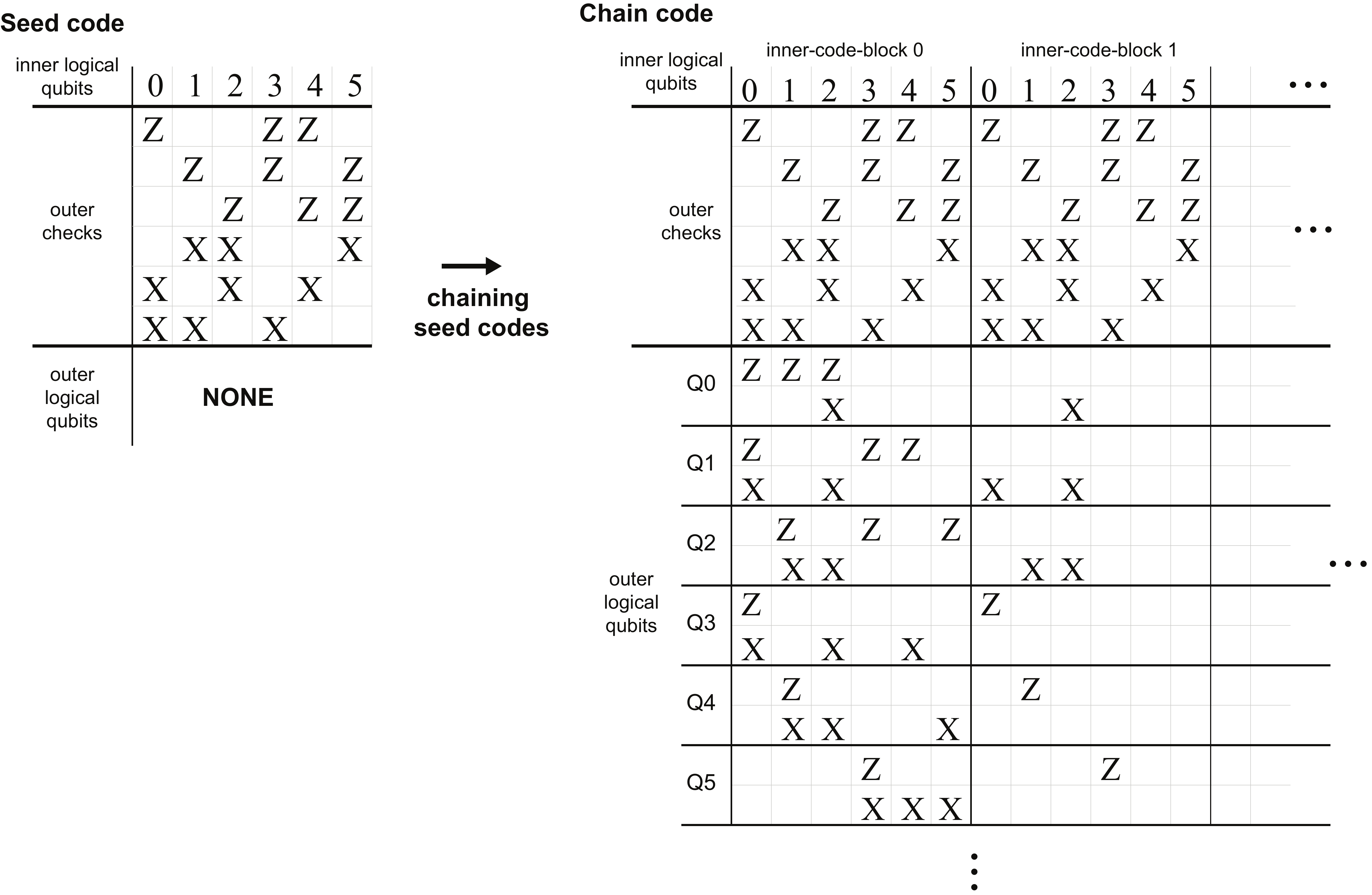}
  \caption{From left to right: the seed code and the chain code.}
  \label{fig:chain_codes}
\end{figure}

\subsection{Workspace Design and Yoking Interval}\label{sec:patch_allocation}
As in the yoked surface code, the multi-yoked surface code reserves workspaces in which the lattice surgeries for yoking are performed; no logical qubits are encoded there. A workspace admits two designs, a time-efficient one and a space-efficient one, both shown in Fig.~\ref{fig:yoking_config}, which depicts the defect diagram during yoking as seen from the side. Which version to use for each check depends strongly on the outer code design; we describe our choices in Sec.\,\ref{sec:overall}.
\begin{figure}[h]
  \centering
  \includegraphics[width=0.45\textwidth]{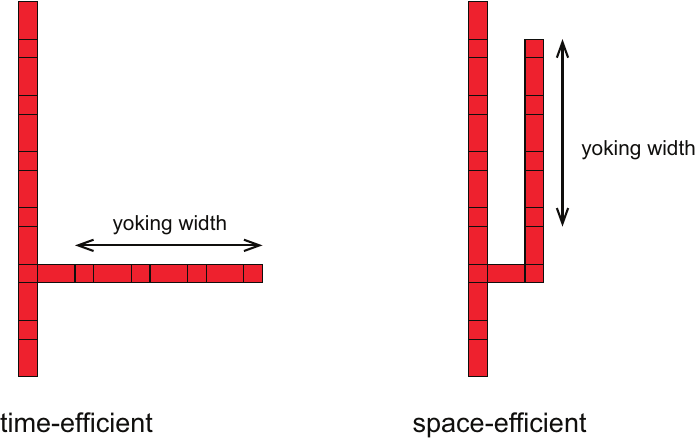}
  \caption{The time-efficient and space-efficient versions of yoking designs.}
  \label{fig:yoking_config}
\end{figure}
In the time-efficient version, the CNOT scheduling of the surface code must be designed carefully, because this lattice surgery presents vertical and horizontal X (Z) boundaries at the same time, which may give rise to hook errors in the surface code. As shown in our previous work~\cite{hirai2026no}, however, such hook errors can be avoided in a hex-grid implementation. 

The surface-code patches are allocated as shown in Fig.~\ref{fig:patch_alloc_example}. Several workspaces are placed to perform yoking, i.e., parity checks for outer codes with lattice surgery. Since most patches are not adjacent to a workspace, the workspaces move up and down through the memory-space to reach them. To keep the memory dense, the number and size of workspaces should be as small as possible, which means that not all patches can be yoked simultaneously; instead, each patch is yoked periodically. If the width of the yoking is too small, a hook error on outer codes can occur during the lattice surgery for outer-code checks, which can reduce the fault distance of outer codes. To avoid this problem, choosing a yoking width of $d_\text{inner} d_\text{outer}$ is sufficient, as logical errors with weight $ d_\text {inner} d_\text{outer}$ are acceptable regarding the target logical error rates. However, this is only a sufficient condition, and we can further reduce the width if the outer codes are designed to tolerate hook errors. In fact, as we will explain in Sec.\,\ref{sec:overall}, we find a construction in which a width-$4d_\text{inner}$ workspace is sufficient to maintain the fault distance $d_\text{inner} d_\text{outer}$ for the obtained chain code when $d_\text{outer} = 8$.

The yoking intervals must also be designed carefully to achieve a small logical error rate while keeping the workspaces relatively small. Furthermore, since the workspaces themselves move up and down, the times at which two of them meet must also be scheduled carefully, so that their checks do not interfere with each other. At the same time, all checks at the same level should share a common yoking interval, because a single check with a long yoking interval degrades the logical error rate even when all the other checks at that level have short intervals. Since the concatenated code has a hierarchical structure, we can design the multi-yoked surface code so that the lower-level checks are performed more frequently and the higher-level ones less frequently. This flexibility in the check schedule is key to significantly reducing the logical error rate of multi-yoked surface codes, which in turn improves the density.

\begin{figure}[h]
  \centering
  \includegraphics[width=0.45\textwidth]{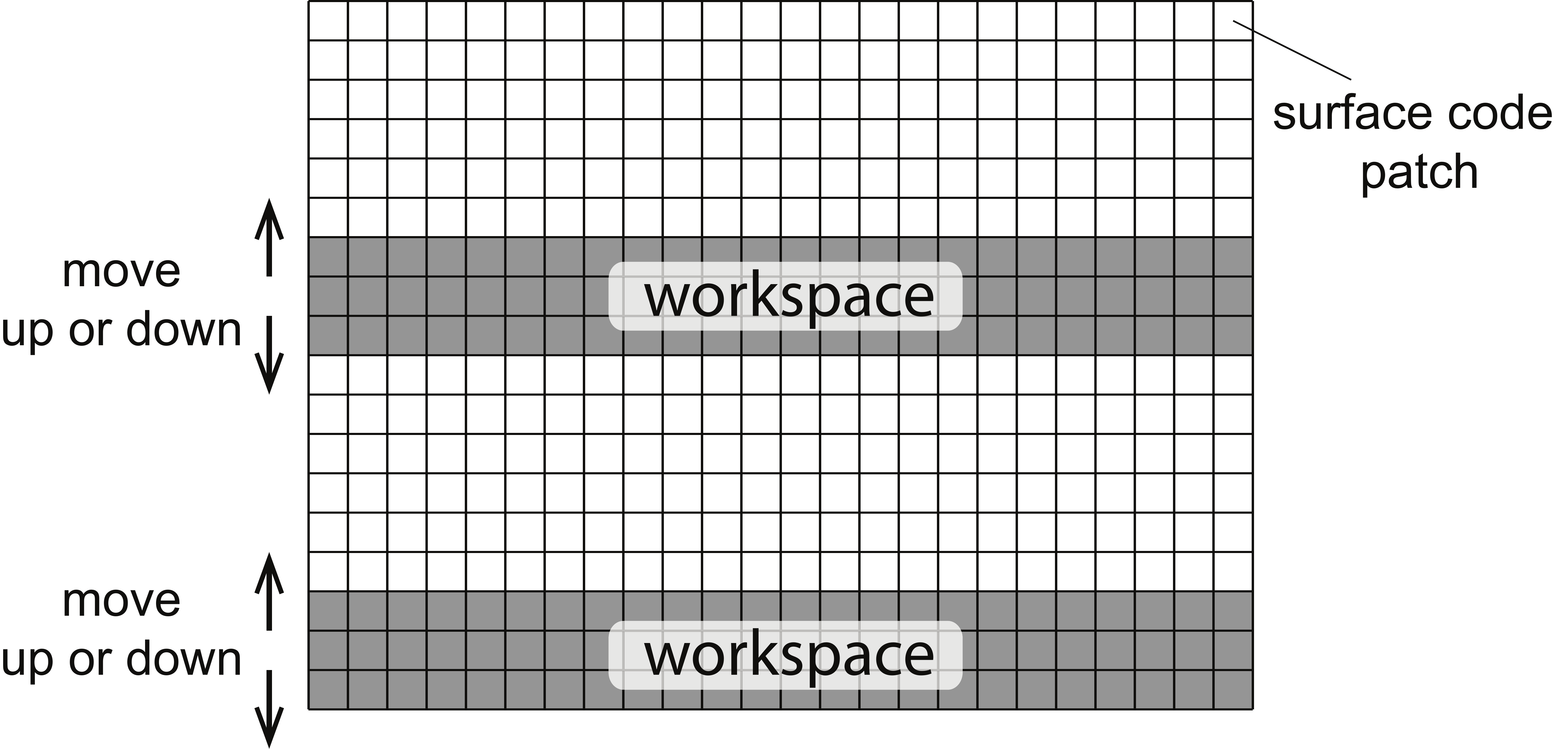}
  \caption{Example of the memory-space layout with the working space.}
  \label{fig:patch_alloc_example}
\end{figure}

\subsection{Yoking Weight Reduction}\label{sec:weight_reduction}
The X and Z parity check matrices for concatenated codes designed with the framework of chain codes typically have high-weight checks, i.e., each check requires lattice surgery operations acting on many surface-code qubits. This is not preferable for two reasons. First, many-body lattice surgery for yoking induces many types of hook errors, which can prevent us from setting the workspace width smaller than $d_\text{inner} d_\text{outer}$. Second, many high-weight checks typically make the outer-code decoding graph have high-weight hyperedges. This degrades the throughput and performance of decoding algorithms compatible with non-matchable decoding graphs, such as Tesseract~\cite{beni2025tesseract}. To mitigate these problems, this section introduces a technique to reduce the weight of a given chain code.

Yoking weight reduction is a technique that heuristically reduces the weight of checks for concatenated codes by modifying the check matrix without changing the encoding space. We apply the following mapping at each level:
\begin{align}
  H'_\text{L1} &= H_\text{L1},\nonumber\\
  H'_\text{L2} &= A H_\text{L2} + B H_\text{L1},\\
  H'_\text{L3} &= C H_\text{L3} + D H_\text{L2} + E H_\text{L1},\nonumber
\end{align}
where $H_{\text{L}i}$ is the check matrix of level $i$, the matrices $A$, $B$, $C$, $D$ and $E$ are binary, and $A$ and $C$ are invertible. This induces the following mapping on the entire check matrix:
\begin{align}
  H' = MH = 
  \begin{bmatrix}
    I & O & O\\
    B & A & O\\
    E & D & C
  \end{bmatrix}\begin{bmatrix}
    H_\text{L1}\\
    H_\text{L2}\\
    H_\text{L3}
  \end{bmatrix},
\end{align}
where $I$ is the identity matrix. Since $I$, $A$, and $C$ are all invertible, the mapping $M$ is invertible as well, and applying it therefore leaves the codespace unchanged while potentially reducing the check weights.

 We searched for such $M$ using LLM agents, and we numerically verified that the resulting code retains the distance $d_\text{inner} d_\text{outer}$ with verified codes. As this procedure is complicated, our implementation and its explanation are available in a GitHub repository~\cite{github}.

\subsection{Overall Design and Parameter Exploration of Multi-Yoked Surface Codes}
\label{sec:overall}
Finally, in this subsection we present the overall design of one example of a multi-yoked surface code memory. We first fix the allocation of the patches and workspaces together with the number of checks at each level of yoking, and then construct a good outer code from these constraints. In this construction, we completely ignore the accessibility of the memory, so additional resources are required when loading the logical information. The cost for loading operations is discussed in Sec.\,\ref{sec:discussion}.

We allocate patches and workspaces as shown in Fig.~\ref{fig:patch_alloc_employed}. The memory space is separated into $m$ types of modules, as shown in the right panel. There are two types of workspaces: narrow workspaces and wide workspaces. Each module contains one wide workspace and $u$ narrow workspaces. Each white square in the figure is a surface-code patch of distance $d_{\rm inner}$. All module types share the patch configuration shown in the left panel, but differ in the check configurations of the concatenated code. For the outer code, we assume that each level-$(\ell+1)$ block takes $n_\ell$ level-$\ell$ blocks as input ($\ell = 0, 1, 2$).

\begin{figure}[t]
  \centering
  \includegraphics[width=0.50\textwidth]{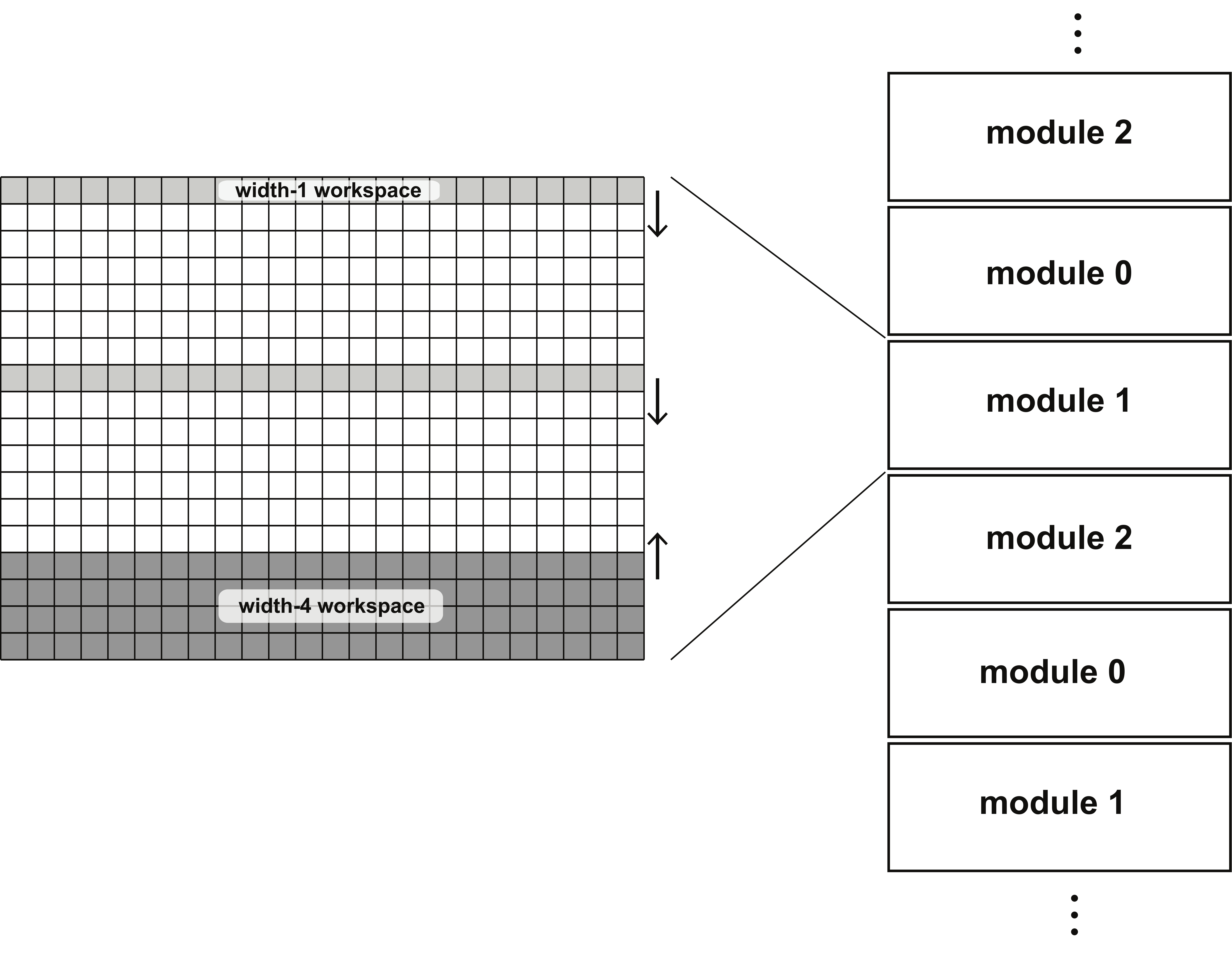}
  \caption{The employed memory-space layout for a good density with parameters $m=3$ types of modules, $u=2$ narrow workspaces per module, width between workspaces $v(s+1)=6$, and a wide workspace with $d_{\rm hook}=4$.}
  \label{fig:patch_alloc_employed}
\end{figure}

Level-1 yoking is performed with the space-efficient version using a narrow workspace. We require that hook-induced logical errors do not occur at distances below $d_\text{hook}d_\text{inner}$, where $d_{\rm hook}$ should be determined so as not to reduce the fault distance of the outer code. This means each check with space-efficient lattice surgery takes $(d_\text{hook}+1)d_\text{inner}$ rounds. Since the level-1 code is an iceberg code, a single level-1 yoking requires $t_\text{narrow} = (2d_\text{hook}+4)d_\text{inner}$ rounds, which consist of $2(d_\text{hook}+1)d_\text{inner}$ rounds for the two yoking checks and $2d_\text{inner}$ rounds of patch movement. This process is shown as a diagram at the top of Fig.\,\ref{fig:yoking_config_employed} (for $d_\text{hook}=4$). As shown in this figure, each narrow workspace moves down every $t_\text{narrow}$ rounds.

Level-2 and level-3 yokings are performed in the time-efficient version using a wide workspace, whose width is $d_\text{hook}d_\text{inner}$. We assume that the wide workspace takes $s$ times as long as the narrow workspace for the yokings and the movements, so $t_\text{wide} = st_\text{narrow}$. To synchronize yoking in the wide workspaces with the movement of the narrow workspace, we choose $s$ to be an integer. As $2d_\text{inner}$ rounds are consumed for patch movement, level-2 and level-3 yokings can use $t_{2,3} = st_\text{narrow}-2d_\text{inner}$ rounds. This process is shown as a diagram at the bottom of Fig.\,\ref{fig:yoking_config_employed} (for $d_\text{hook}=4$ and $s = 2$). As shown in this figure, each wide workspace moves up every $t_\text{wide}$ rounds.

\begin{figure}[t]
  \centering
  \includegraphics[width=0.40\textwidth]{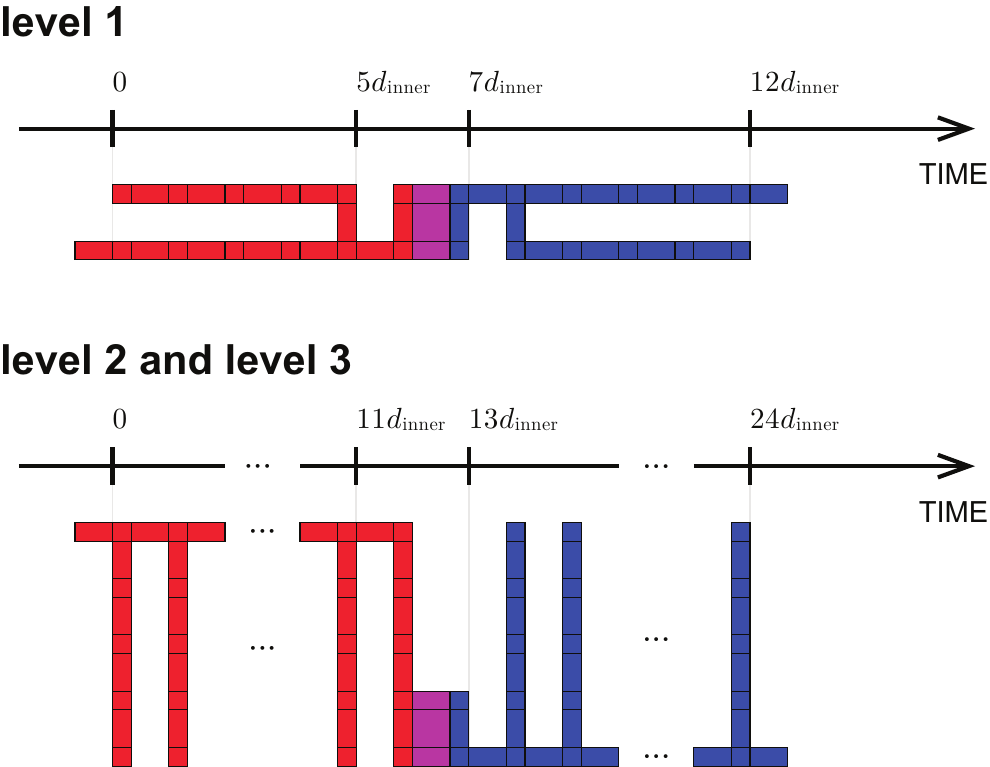}
  \caption{The time evolutions of level-1, level-2 and level-3 checks.}
  \label{fig:yoking_config_employed}
\end{figure}

Level-2 checks are performed each time a row of surface code patches meets a wide workspace. In contrast, level-3 checks are shared among $m$ modules, i.e., level-3 checks are performed $m$ times less frequently than level-2 checks. Therefore, if we denote the number of checks for level-2 and level-3 codes by $c_2$ and $c_3$, respectively, they must satisfy $c_2 + \lceil c_3/m \rceil \leq t_{2,3}/d_\text{inner}$. Once we determine the iceberg code size $n_0$ and its block count per level-2 chain code $n_1$, the number of level-2 checks $c_2$ is bounded by Proposition\,\ref{prop:minimum_checks}. This also bounds $c_3$ as $\lceil c_3/m \rceil \leq t_{2,3}/d_\text{inner} - c_2$. The level-3 seed code is explored using the subroutine explained in Sec.\,\ref{sec:level3-code} with $c_3$ checks (i.e., $c_3/2$ each for X and Z) as an input parameter.

Every $t_\text{wide}$ rounds, each narrow workspace moves down $s$ rows, and each wide workspace moves up one row. To avoid the narrow workspace moving down during the yoking of the wide workspace, we set the distance between the two workspaces to $v(s+1)$ rows, where $v$ is an integer, so that the two types always meet right after exactly one set has completed all of its checks. See Fig.\,\ref{fig:meanders} for an example of a full yoking schedule.
After several cycles, workspaces eventually reach the top or bottom of the memory space. Then, we swap the roles of the wide workspace at the top and the narrow workspace at the bottom by using free qubits, as shown in Fig.\,\ref{fig:swap_wide_and_narrow}.

\begin{figure}[h]
  \centering
  \includegraphics[width=0.60\textwidth, angle=0]{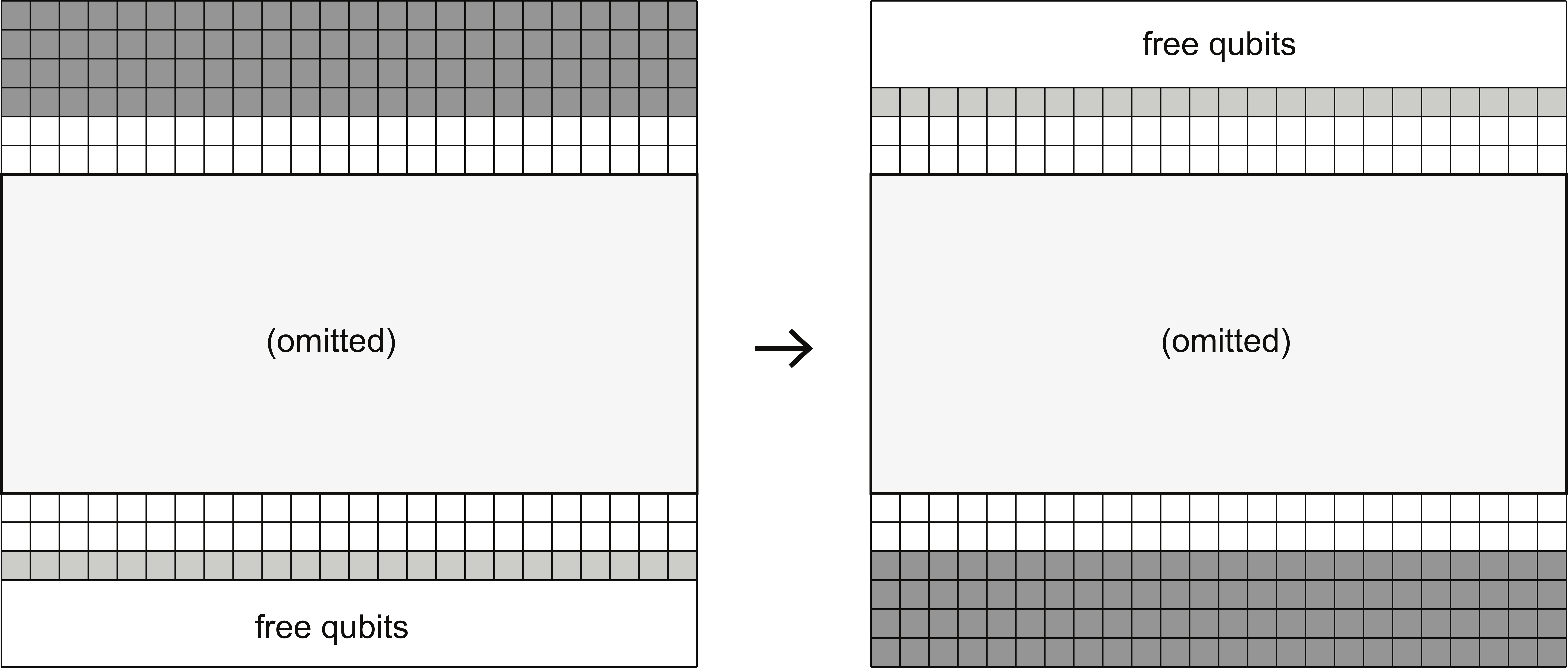}
  \caption{Snapshots of the moment when the workspaces reach the top and bottom ends of the memory-space and their roles are swapped using free qubits.}
  \label{fig:swap_wide_and_narrow}
\end{figure}

With this design, the level-1 yoking is performed every $t_1 = v(s+1)t_\text{narrow}$ rounds. The level-2 yoking is performed every $t_2 = sut_1$ rounds, and the level-3 yoking is performed every $t_3 = msut_1$ rounds. Thus, the ratio of yoking intervals can be tuned with the parameters $s$, $u$ and $m$.

In summary, there are several design parameters as listed in Table\,\ref{tab:design_parameters}. Although we will sweep these parameters to examine the parameter dependency of the multi-yoked surface codes, we list the best-performing set in the teraquop regime: inner-code distance $d_{\rm inner}=10$, $n_0=16$, $n_1=8$, $n_2=2$, $u=2$, $v=2$, $s=2$, $d_{\rm hook}=4$, and $m=3$. Fig.~\ref{fig:patch_alloc_employed} shows an example with these parameters.
These parameters determine the other parameters as follows. As we see in Sec.\,\ref{sec:level2-code}, the minimum and achievable seed-code size for $[[16,14,2]]$ iceberg code is $c_2 = 8$. Thus, the level-2 chain code becomes $[[112,104,2]]$. This requires $c_3$ to satisfy $c_3 \leq 42$. We run the subroutine in Sec.\,\ref{sec:level3-code} with this restriction, and the subroutine outputs a $[104, 83, 7]$ classical code as an X and Z seed code (21 checks for Z and X each, and thus $c_3=42$). Thus, the level-3 chain code is $[[208,166,2]]$. Using our heuristic search combined with the yoking weight reduction in Sec.\,\ref{sec:weight_reduction}, we found a transformation of codes under which the outer code tolerates $4 d_{\rm inner}$-weight hook errors for outer checks. Thus, the choice of $d_{\rm hook}=4$ does not degrade the total fault distance of the outer code, even though $d_{\rm hook}=4$ is smaller than $d_{\rm outer}=8$. The final code structure is shown in Fig.~\ref{fig:the_multi_yoked_surface_code}.

\begin{table}[hb]
    \centering
    \begin{tabular}{|c|c|}
    \hline
    Symbol & Description \\
    \hline
    \hline
    $d_{\rm inner}$ & Code distance of level-0 surface codes \\
    \hline
    $n_0$ & Number of level-0 (surface-code) blocks per level-1 iceberg code $[[n_0,n_0-2,2]]$ \\
    \hline
    $n_1$ & Number of level-1 code blocks per level-2 chain code \\
    \hline
    $n_2$ & Number of level-2 code blocks per level-3 chain code \\
    \hline
    $u$ & Number of narrow workspaces per module \\
    \hline
    $s$ & Yoking time for wide workspace relative to that for narrow workspace \\
    \hline
    $v$ & Distance between two workspaces divided by $(s+1)$ \\
    \hline
    $d_{\rm hook}$ & Width of wide workspace divided by $d_{\rm inner}$ \\
    \hline
    $m$ & Number of module types in memory space \\
    \hline
    \end{tabular}
    \caption{Design parameters of multi-yoked surface codes}
    \label{tab:design_parameters}
\end{table}

\begin{figure}[h]
  \centering
  \includegraphics[width=0.95\textwidth, angle=0]{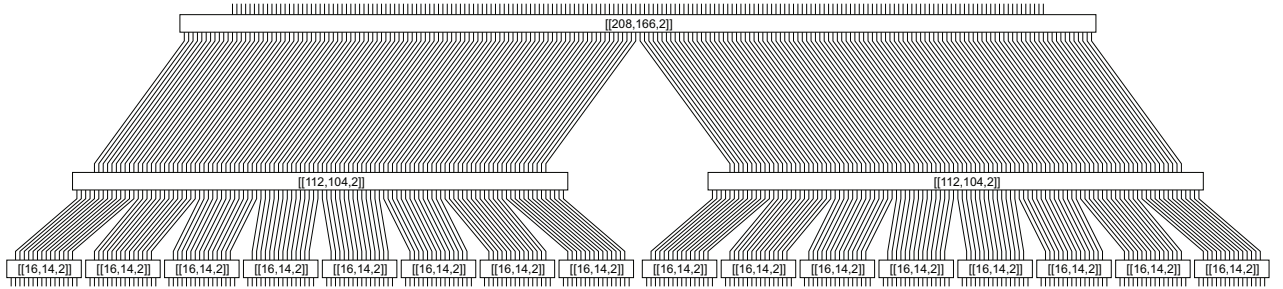}
  \caption{The densest design of multi-yoked surface codes in this work.}
  \label{fig:the_multi_yoked_surface_code}
\end{figure}

\section{Numerical Results}
\label{sec:numerical_results}

\subsection{Evaluation Methods and Approximations} \label{sec:numerical_setting}
We have evaluated the performance of the proposed memory-space implementation. The full circuit is too large to simulate directly in Stim~\cite{gidney2021stim}, so we utilized the following evaluation methods and approximations. Hereafter, we denote the surface code distance by $d$.

The outer codes in the multi-yoked surface codes are not matchable, and we cannot handle the whole decoding task with fast MWPM decoders. Thus, the inner and outer codes are simulated separately. We first evaluate the logical error events of the idling circuits of the inner surface codes. Then, the outer code simulation is executed with gap values pre-sampled from them and decoded with Tesseract~\cite{beni2025tesseract}. This approximation simplifies decoding for the outer codes and is necessary to make the evaluation feasible with our computational resources.

In the evaluation of outer codes, we ignore the following errors in outer codes that are not expected to contribute significantly to the final logical error rates, similar to Ref.~\cite{gidney2025yoked}. 
The hook errors occurring in the outer-code checks are ignored because we allocate sufficient workspace width and the number of logical error events involving such errors is negligible. 
Since the time-like edges corresponding to the measurement errors did not contribute significantly to the logical error rates as discussed in Ref.~\cite{gidney2025yoked}, we assume perfect syndrome extraction for outer checks. 
Spacetime-like edges arising from the different schedules of the lattice surgery operations are ignored in the gap simulation. We expect this to have a negligible effect on the final logical error rates, since such error events are very rare. 

We utilized Monte Carlo sampling for evaluating logical error rates. As the evaluation for logical error rates in the teraquop regime requires massive sampling, we utilized several computational resources in parallel. We used a workstation with two Intel Xeon Platinum 8173M CPUs (56 cores, 112 threads, 2.0 GHz base, up to 3.5 GHz) with 256 GB memory, and 22 compute nodes of the RIKEN R-CCS Cloud: 14 nodes, each with an AMD EPYC 9684X CPU (96 cores, 192 threads, 2.55 GHz base, up to 3.7 GHz) and 768 GB memory; 2 nodes, each with two Intel Xeon Platinum 8470 CPUs (104 cores, 208 threads, 2.0 GHz base, up to 3.8 GHz) and 2,048 GB memory; 4 nodes, each with two AMD EPYC 9554 CPUs (128 cores, 256 threads, 3.1 GHz base, up to 3.75 GHz) and 1,536 GB memory; 1 node with two AMD EPYC 9534 CPUs (128 cores, 256 threads, 2.45 GHz base, up to 3.7 GHz) and 1,536 GB memory; and 1 node with an Intel Xeon w9-3575X CPU (44 cores, 88 threads, 2.2 GHz base, up to 4.8 GHz) and 512 GB memory.

\subsection{Evaluation of Inner Codes}
First, we sampled the surface code data. 
The syndrome extraction circuit for $12d$ rounds is constructed using the hex-grid scheduling of Ref.~\cite{hirai2026no} to avoid all hook errors in it. 
We applied the noise model SI1000~\cite{gidney2025yoked} with physical error rate $p=10^{-3}$. The memory circuit does not include errors in reset or measurement operations for data qubits, since multiple copies of this circuit are stacked as building blocks to construct the gap simulation circuit. 
We used Stim~\cite{gidney2021stim} for generating circuits and detector error models, and errors are decoded by PyMatching~\cite{higgott2022pymatching} with two-pass decoding~\cite{fowler2013optimal}. The gap values are then calculated on the decoding graph updated by the results of the first-pass decoding. 

Histograms of sampled gap values are shown in the left panel of Fig.~\ref{fig:gap_distribution} for several inner-code distances. Note that a positive gap value means that the decoder's estimate was correct, and a negative value means that it was incorrect. Similar to the results in Ref.~\cite{gidney2025yoked}, we see the distribution becomes smooth as the code distance increases. 

We then calibrated the gap for the outer decoder by fitting the log-likelihood ratio of an inner-decoder logical error as an affine function of the gap, $\Lambda_{\mathrm{aff}}(g) = \alpha g + \beta$, where $\alpha$ is the scale and $\beta$ is the offset, by minimizing the log-loss over the samples. Such a fit is dominated by the many low-gap samples, so we sorted the samples by gap, cut them into bins each containing a fixed number of observed or affine-predicted logical errors, and in each bin shifted the affine line by the shrunk log ratio of observed to predicted errors. The shifted points were monotonized and linearly interpolated, and the result was adopted only when it beat the affine fit on held-out log-loss. The results are shown in the right panel of Fig.~\ref{fig:gap_distribution}. This figure shows the logical error rates for observed unsigned gap values and their fit. We see that the gaps are well calibrated.

\begin{figure}[h]
  \centering
  \includegraphics[width=0.45\textwidth]{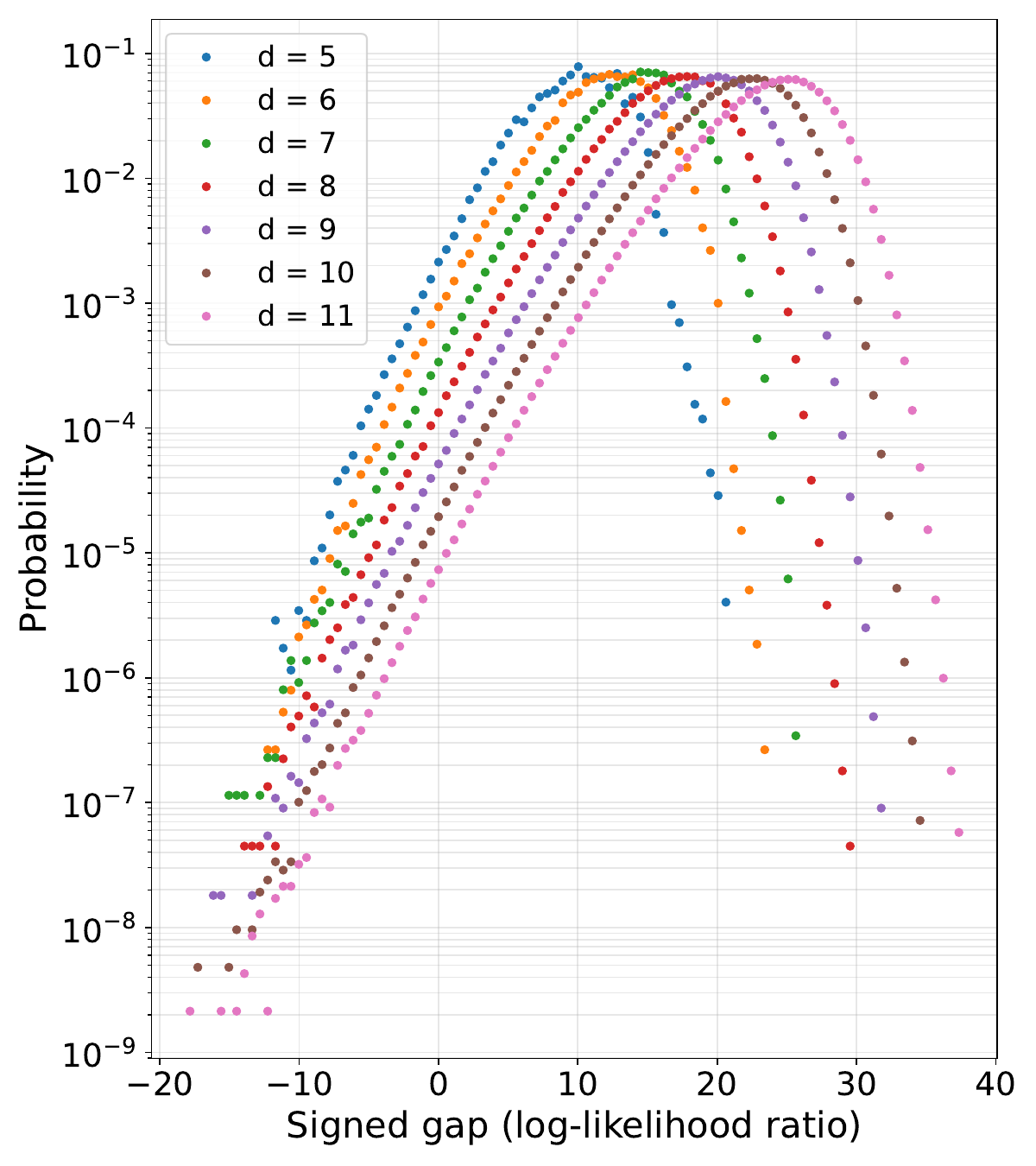}\hspace{40pt}
  \includegraphics[width=0.45\textwidth]{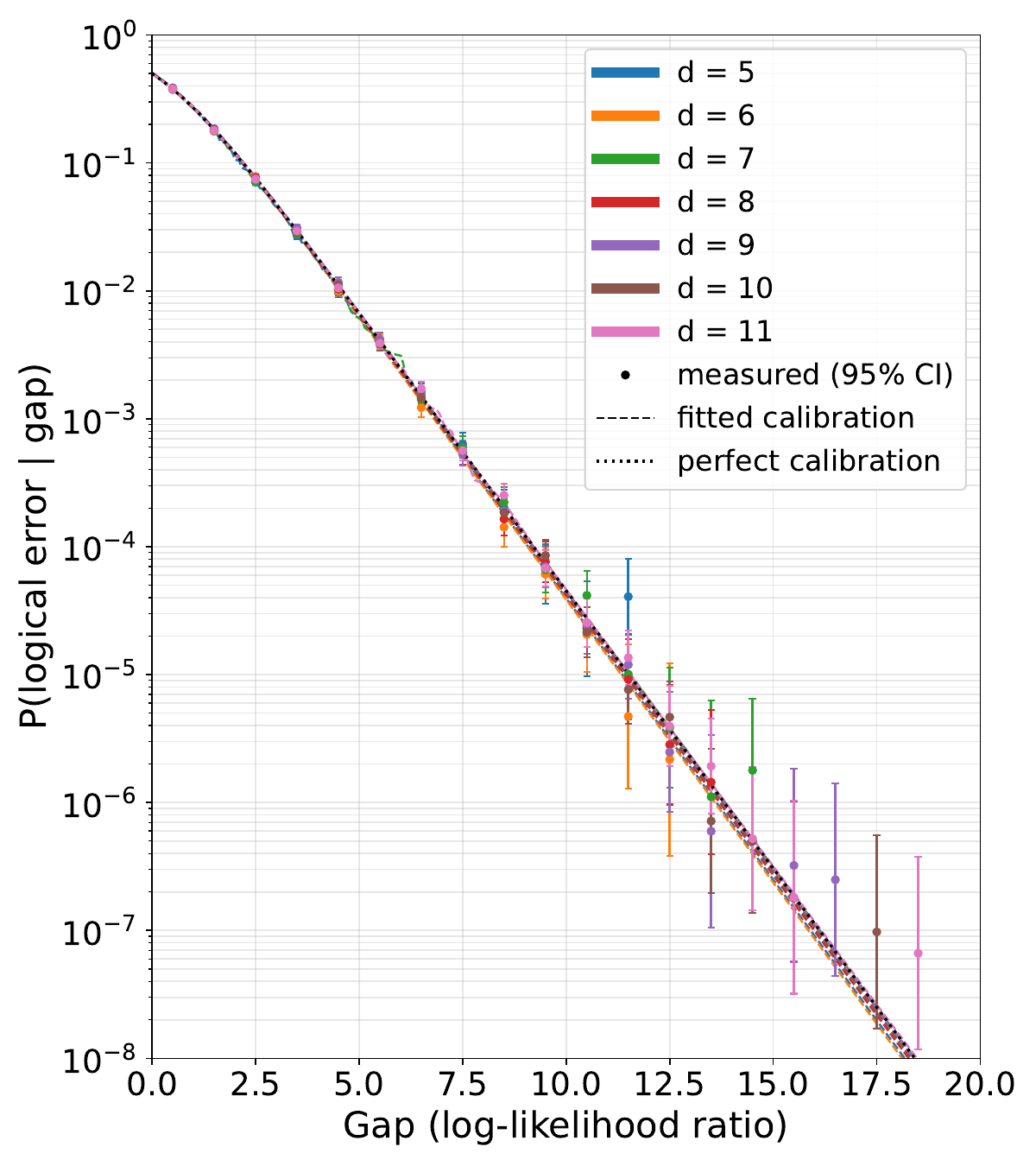}
  \caption{Left: distributions of the sampled gap values. Right: calibration of the gap values. The samples are taken from $12d$ rounds memory experiments, with approximately $10^4$ logical error events collected for each distance $d$.}
  \label{fig:gap_distribution}
\end{figure}

\subsection{Evaluation of Yoked Surface Codes}
Next, we evaluated the logical error rate of the yoked surface codes. We simulated two kinds of yoked surface codes, concatenated with the $[[8,6,2]]$ and $[[16,14,2]]$ codes, respectively, over 10 outer rounds. The yoking intervals are chosen from $12d$, $24d$, $48d$, $72d$, and $96d$. 
For error estimation, we employed the Tesseract decoder~\cite{beni2025tesseract}. We used the gap values pre-sampled in the previous subsection. Note that the outer code is decoded round by round, as we ignored time-like and spacetime-like errors in outer-code decoding as explained in Sec.\,\ref{sec:numerical_setting}. To improve decoding accuracy and throughput, we set its \texttt{det\_order} to the order in which the detectors are measured in the circuit. 

The evaluation results are shown in Fig.~\ref{fig:one_yoke_ler}. The dashed lines show the performance of inner codes for reference, and the solid lines show that of the outer codes. Each color corresponds to an interval value, and the legend lists the slopes for the inner and outer codes. We see that longer intervals significantly worsen the logical error rates of the outer codes. In contrast, the slope is almost independent of the interval values. The logical error rates are exponentially suppressed with the inner code distance $d$, and the slope of the outer codes is approximately double that of the inner code. When we increase the density from $[[8,6,2]]$ to $[[16,14,2]]$, the logical error rates increase. This indicates a trade-off between encoding rate and logical error rate.

\begin{figure}[h]
  \centering
  \includegraphics[width=0.45\textwidth]{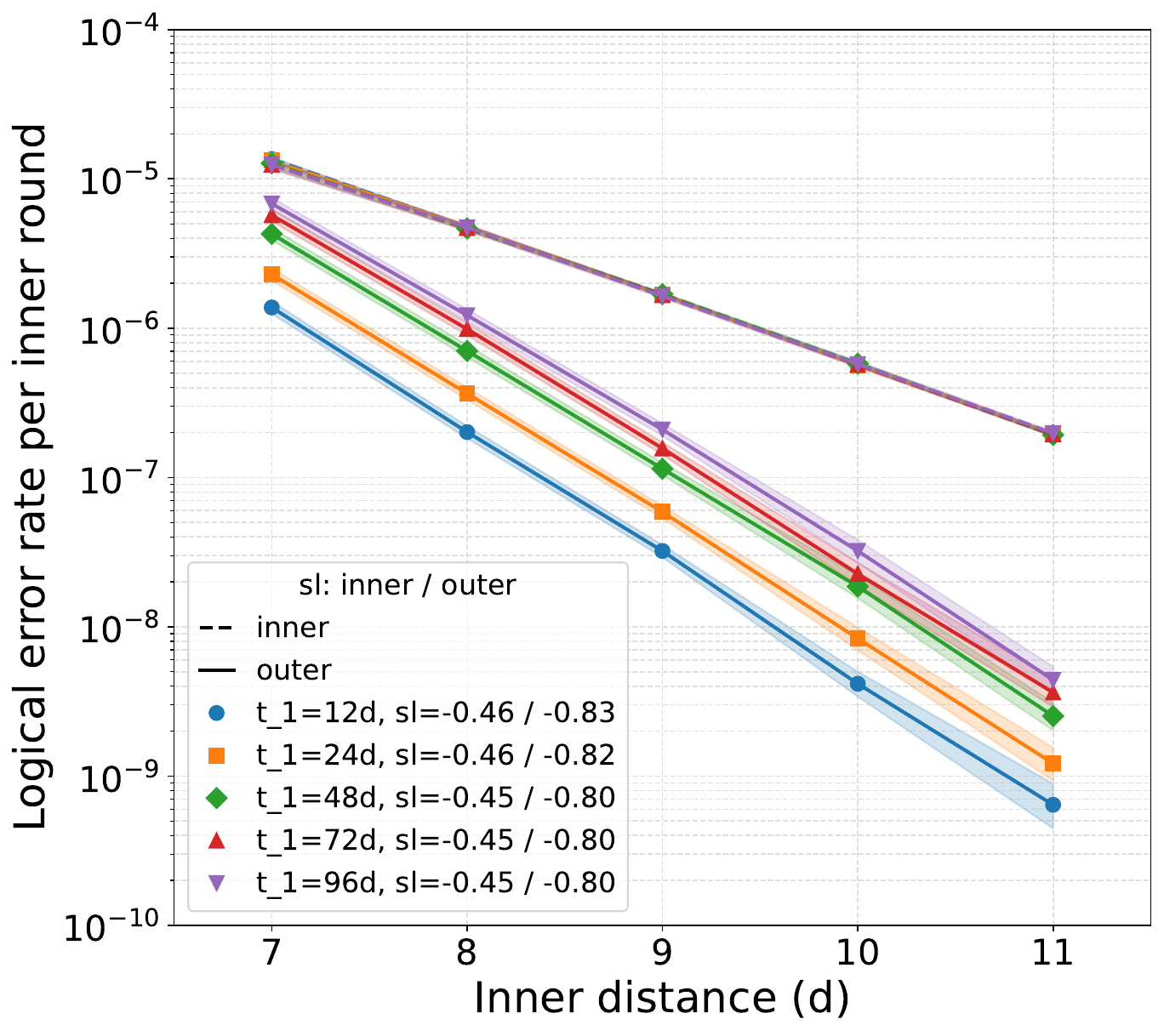}\hspace{30pt}
  \includegraphics[width=0.45\textwidth]{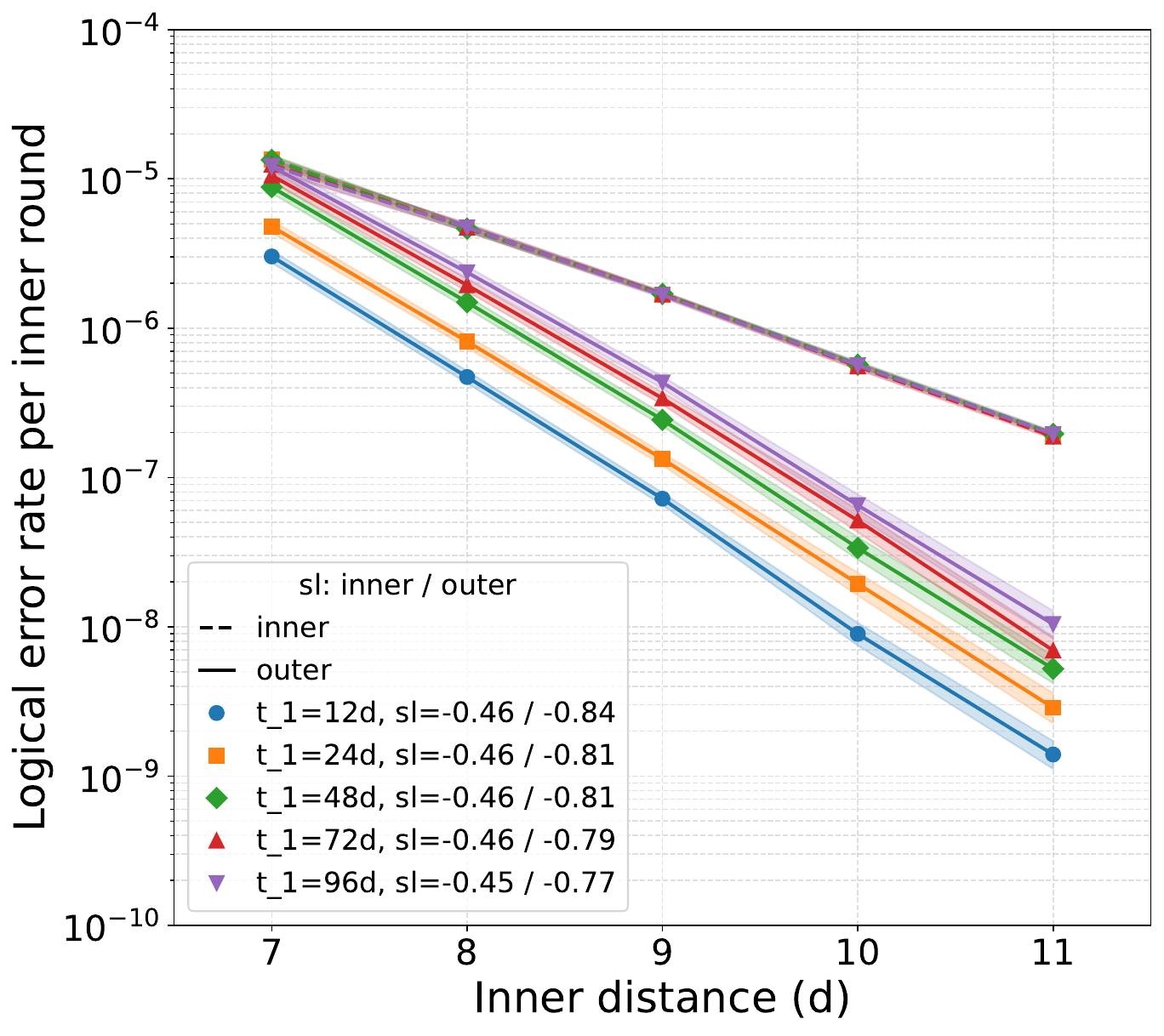}
  \caption{Logical error rates of the yoked surface code with the iceberg codes $[[8,6,2]]$ (left) and $[[16,14,2]]$ (right). \texttt{t\_1} is a yoking interval and \texttt{sl} is a slope. The logical error rates shown are per logical qubit.}
  \label{fig:one_yoke_ler}
\end{figure}

\subsection{Evaluation of Two-Yoked Surface Codes}
We evaluated the logical error rates of the two-yoked surface codes over 10 outer rounds. We again focused on two iceberg codes, $[[8,6,2]]$ and $[[16,14,2]]$, for the level-1 concatenation. To construct level-2 chain codes, we choose the seed codes described in Sec.\,\ref{sec:level2-code}. When we create chain codes from the seed codes, we repeat $r=6$ and $r=8$ iceberg-code blocks, respectively.

The logical error rates and fitted slopes are shown in Fig.~\ref{fig:two_yoke_ler}. The performance of chain codes with six $[[8,6,2]]$ iceberg codes is shown at the top, and that for eight $[[16,14,2]]$ iceberg codes is shown at the bottom. Let $t_1$ and $t_2$ denote the level-1 and level-2 yoking intervals. We swept the level-2 yoking interval $t_2$ over $144d$, $288d$, and $576d$, which corresponds to the column. In each figure, we swept $t_1$. We see that the logical error rates are suppressed exponentially with the inner code distance. The two-yoked surface codes suppress the logical error rates more strongly than the yoked surface codes do, with slopes approximately quadruple those of the inner code. 
We also see that the shorter the level-2 interval becomes, the more the logical error rate is improved by frequent level-1 yoking. We finally conclude that the higher-level interval can be made longer while keeping the logical error rate moderate.

\begin{figure}[t]
  \centering
  \includegraphics[width=0.30\textwidth]{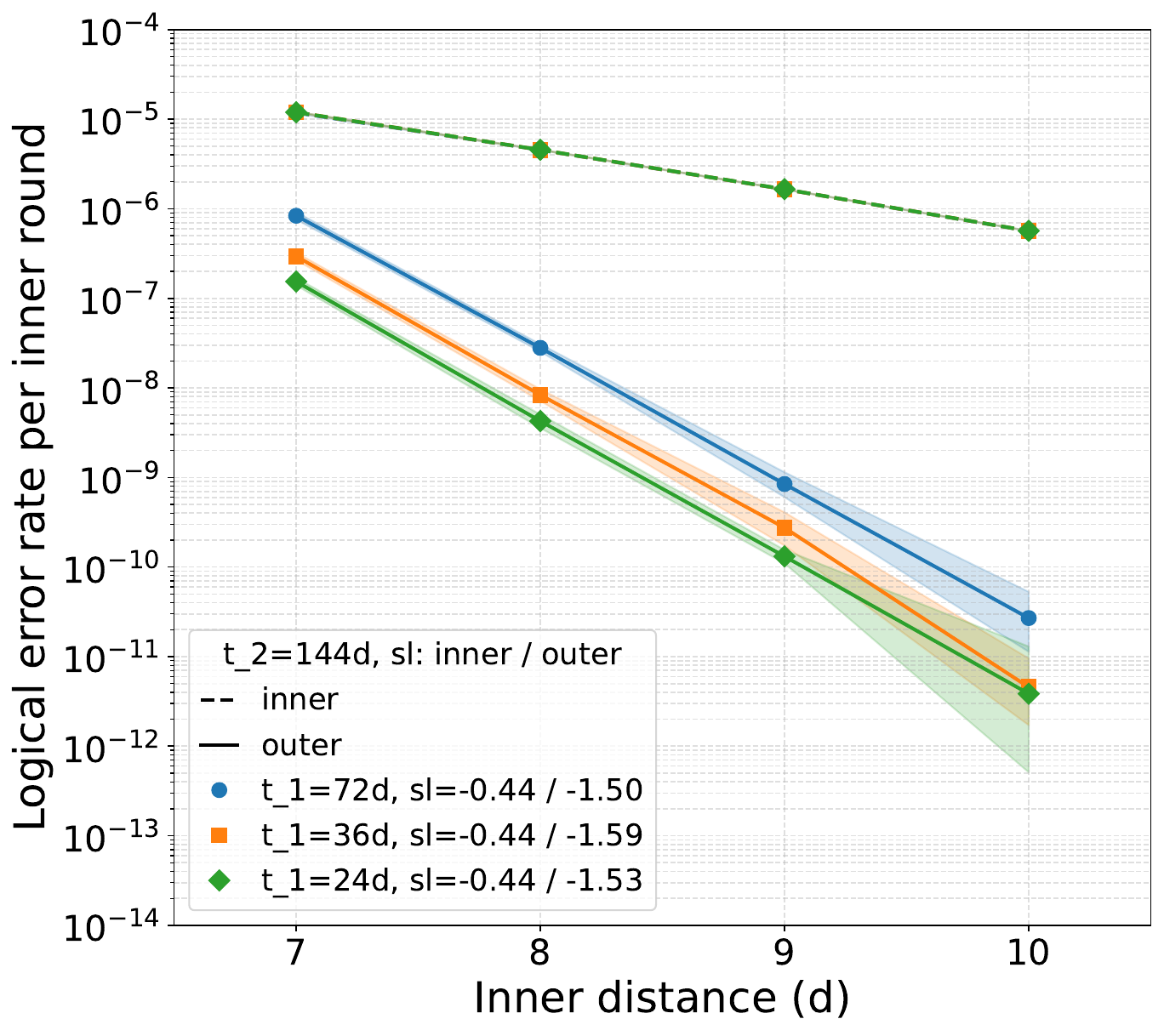}\hspace{10pt}
  \includegraphics[width=0.30\textwidth]{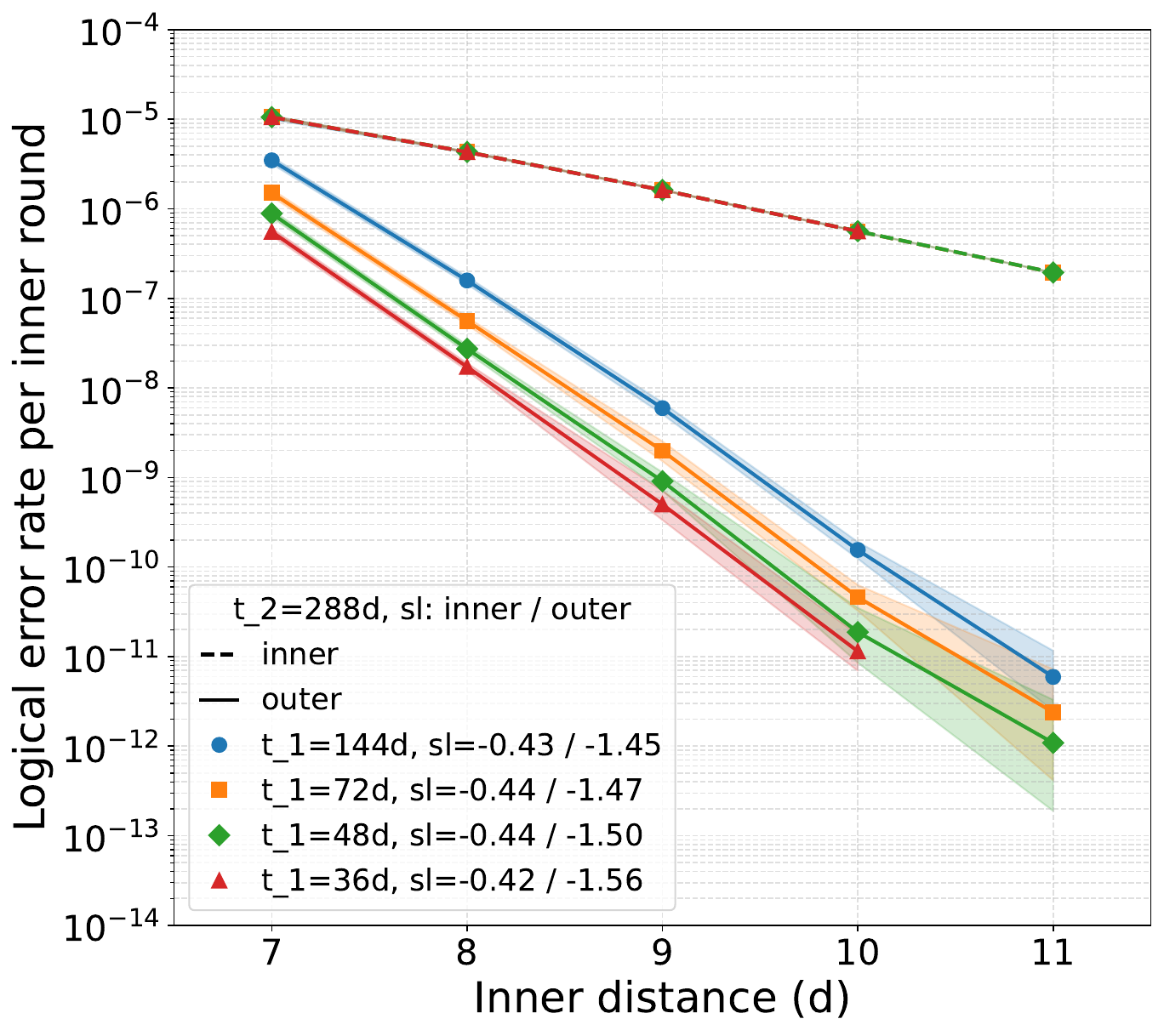}\hspace{10pt}
  \includegraphics[width=0.30\textwidth]{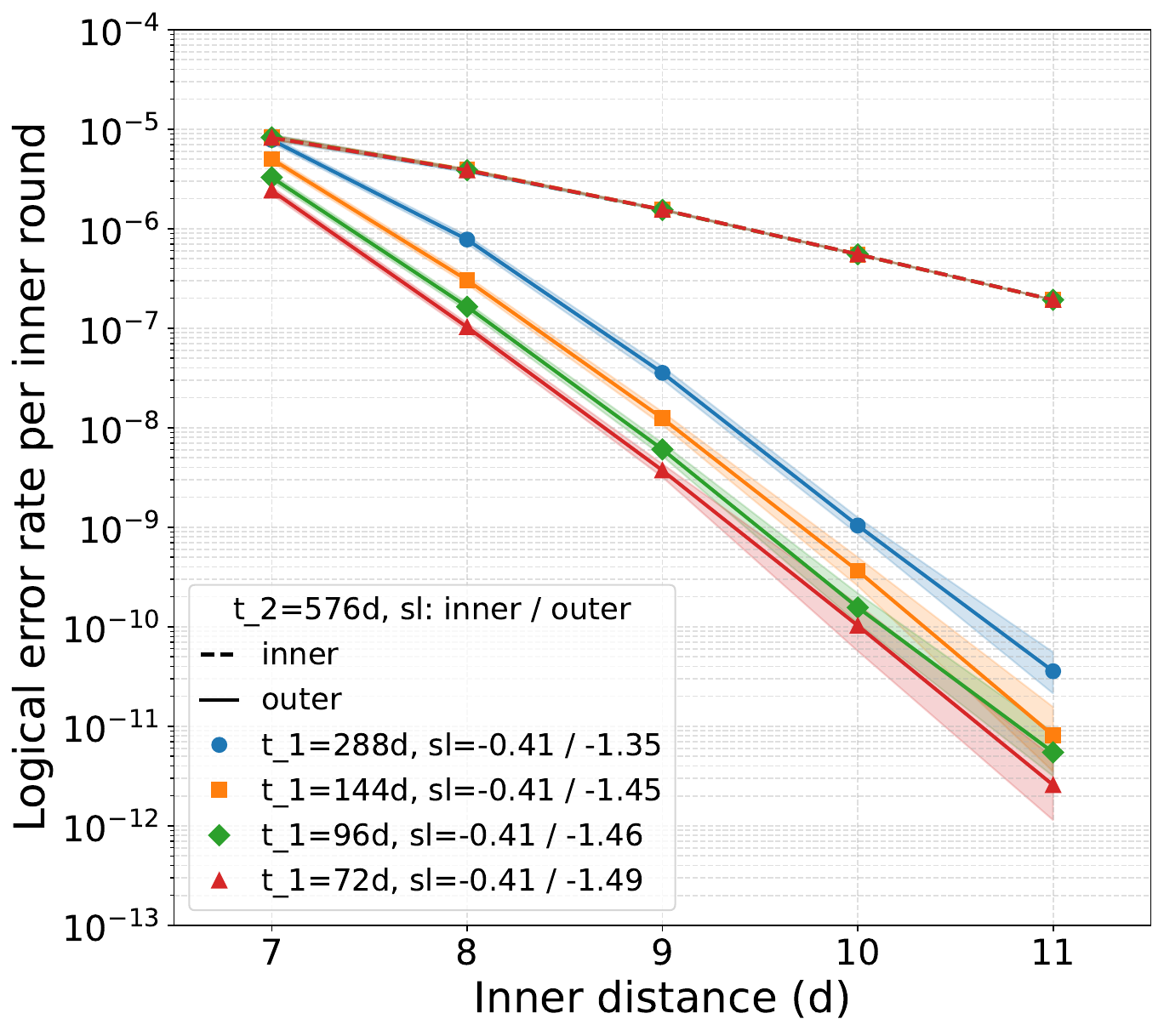}
  \includegraphics[width=0.30\textwidth]{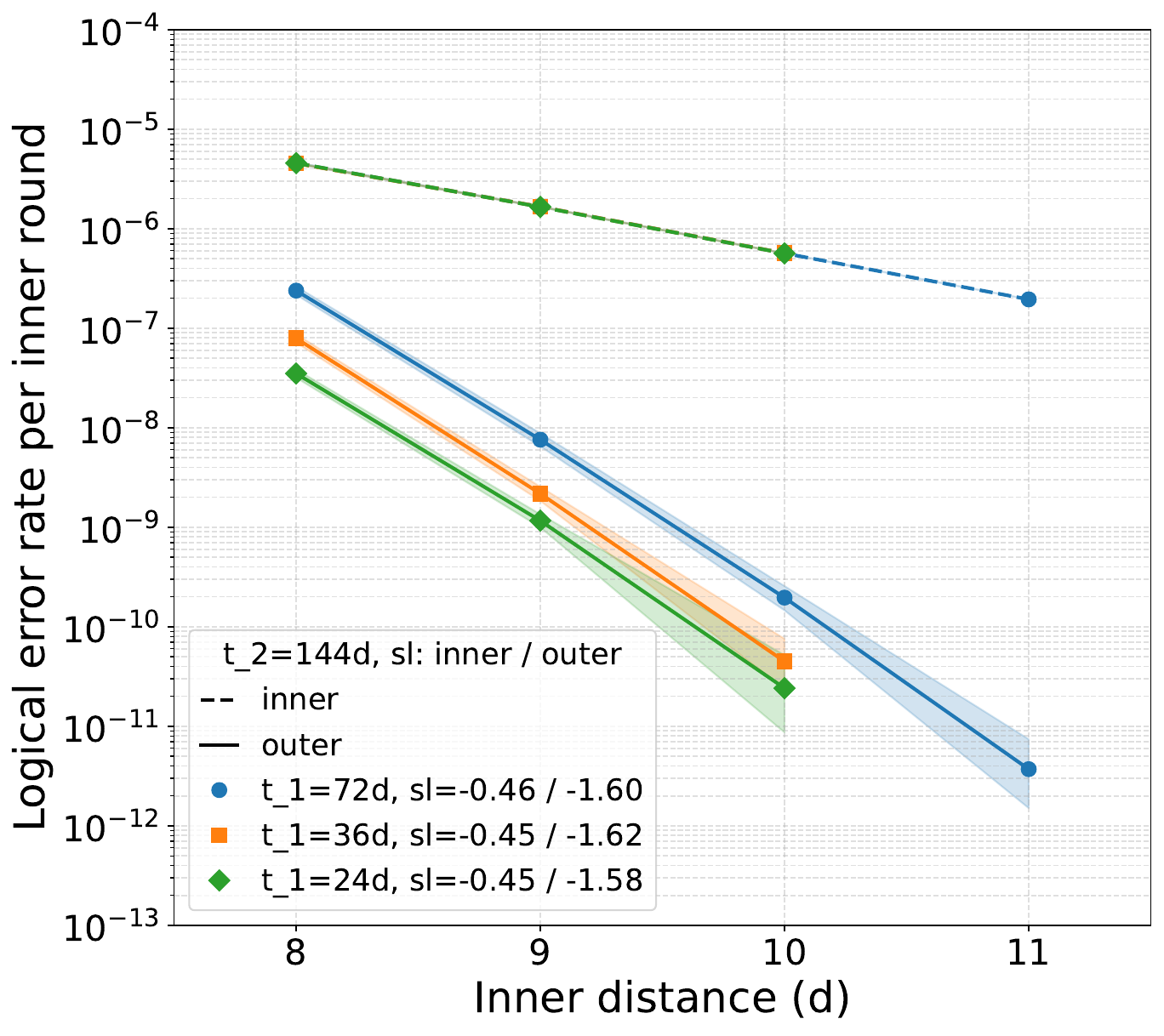}\hspace{10pt}
  \includegraphics[width=0.30\textwidth]{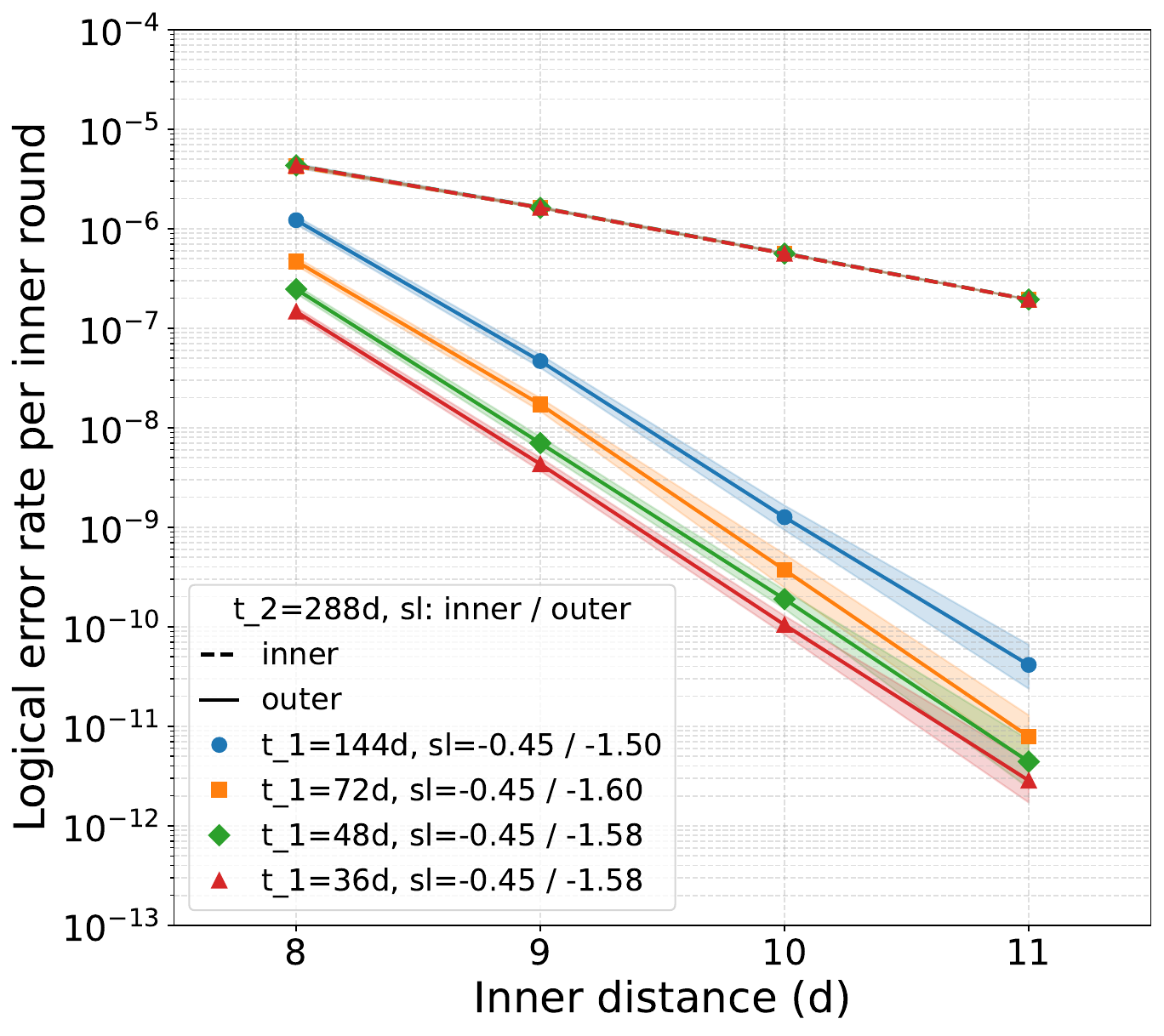}\hspace{10pt}
  \includegraphics[width=0.30\textwidth]{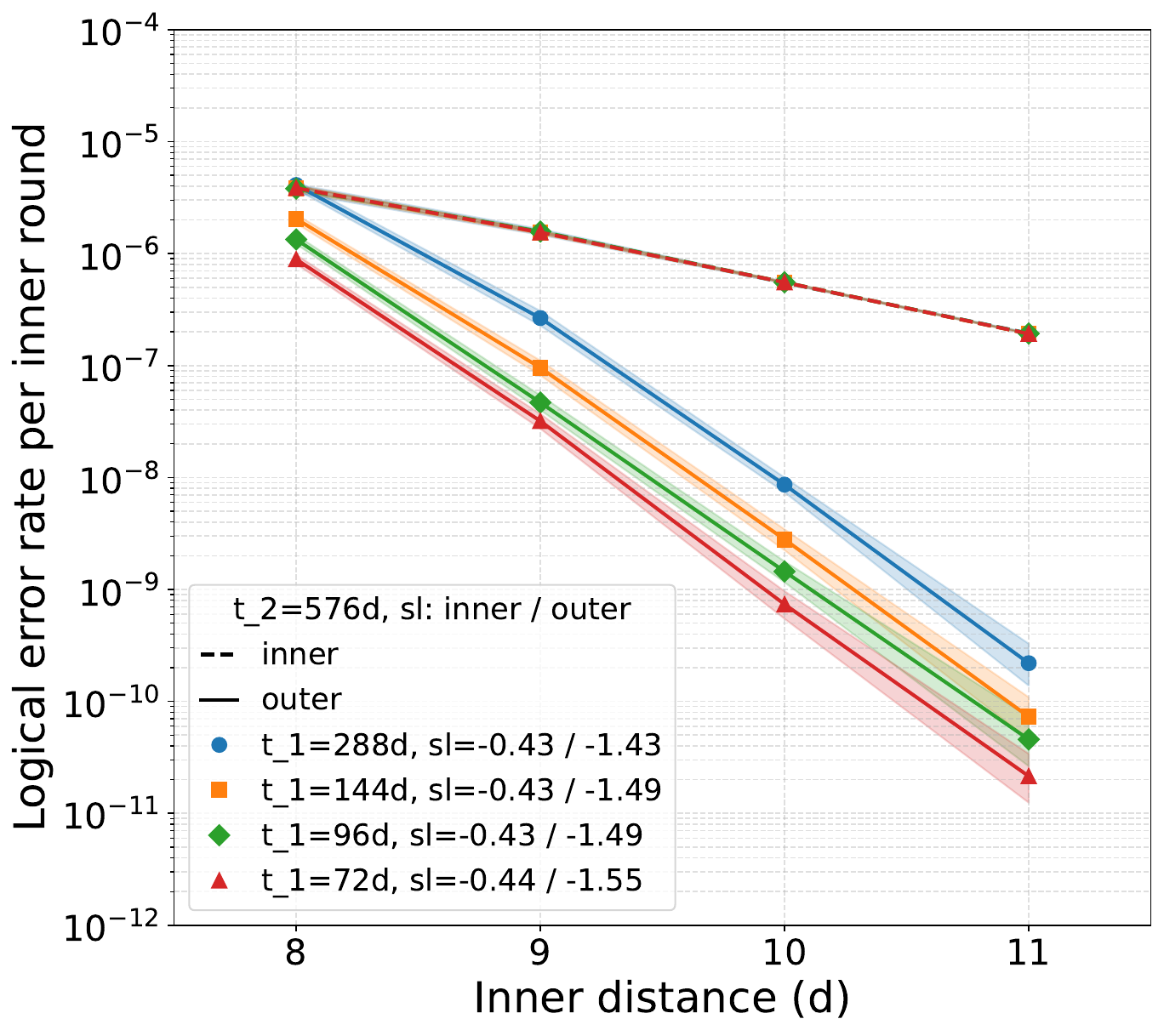}
  \caption{Logical error rates of the two-yoked surface codes for six iceberg codes $[[8,6,2]]$ (top) and for eight iceberg codes $[[16,14,2]]$ (bottom). \texttt{t\_1} is a level-1 yoking interval, \texttt{t\_2} is a level-2 yoking interval and \texttt{sl} is a slope. The logical error rates shown are per logical qubit.}
  \label{fig:two_yoke_ler}
\end{figure}

\subsection{Evaluation of Three-Yoked Surface Codes}
From the results so far, we design the three-yoked surface code, i.e., the densest code in this work, so that the level-2 interval is much longer than the level-1 interval and the level-3 interval is much longer than the level-2 interval. These code choices depend heavily on the surface-code layout explained in the previous section. 

The results of the three-yoked surface code are shown in the left panel of  Fig.~\ref{fig:surface_and_three_yoke_ler}. We simulated the code shown in Fig.~\ref{fig:the_multi_yoked_surface_code} over four outer rounds. We see that the three-yoked surface code also suppresses the logical error rates exponentially. However, its slope of $2.87$ falls short of the ideal value of $0.42\times 8\simeq 3.4$, i.e., eight times the inner-code slope expected from the distance-8 outer code. This may be because the simulated $d$ regime has not yet reached the regime where the slope is maximal. From this discussion, the data point at $d=10$ will most likely lie below the fitted line, but we nevertheless use the fitted line to estimate the density in the next paragraph. We also simulated the hex-grid surface code with the two-pass decoder (the right panel of Fig.~\ref{fig:surface_and_three_yoke_ler}) for the density comparison in the next subsection.

\begin{figure}[h]
  \centering
  \includegraphics[width=0.49\textwidth]{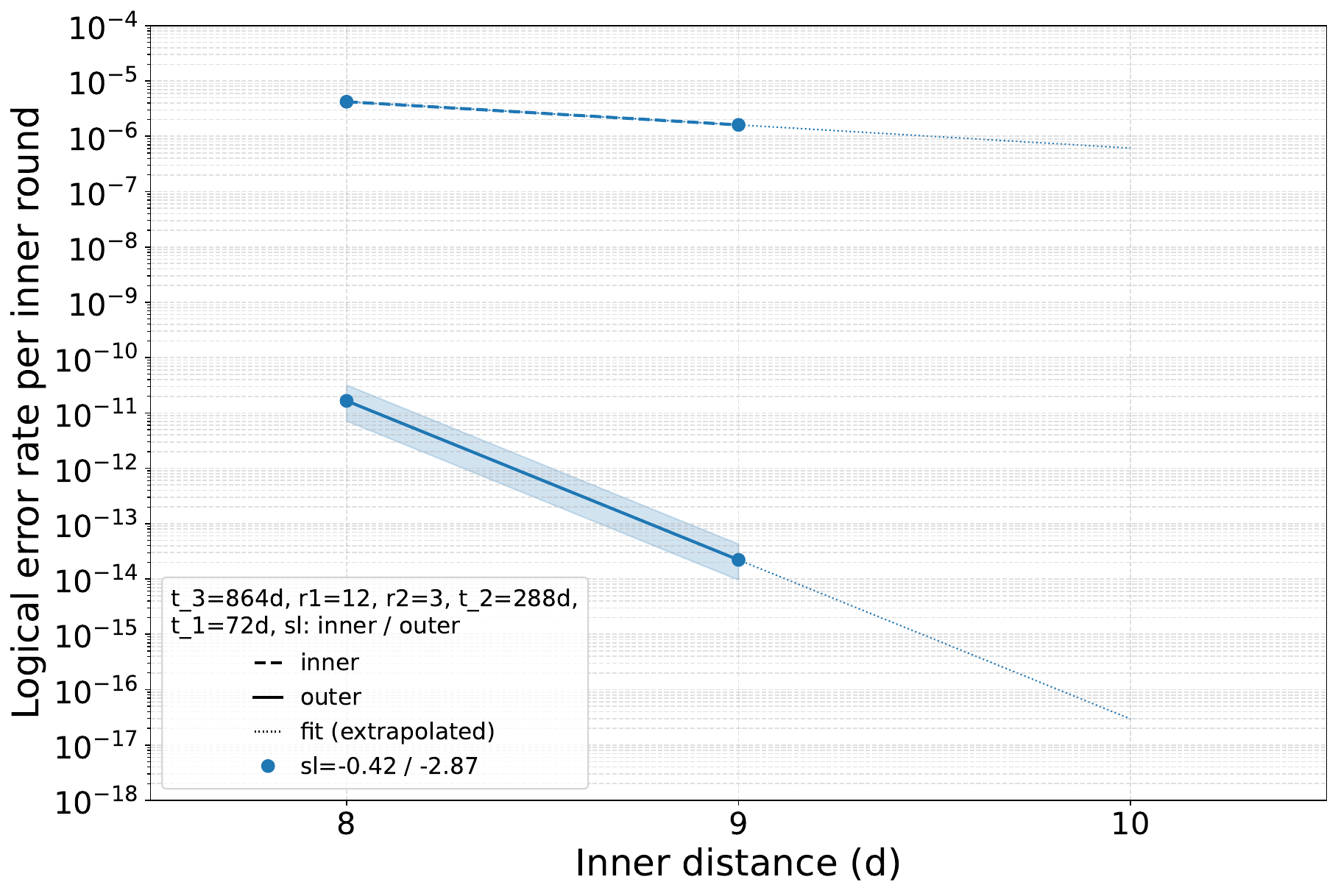}
  \includegraphics[width=0.49\textwidth]{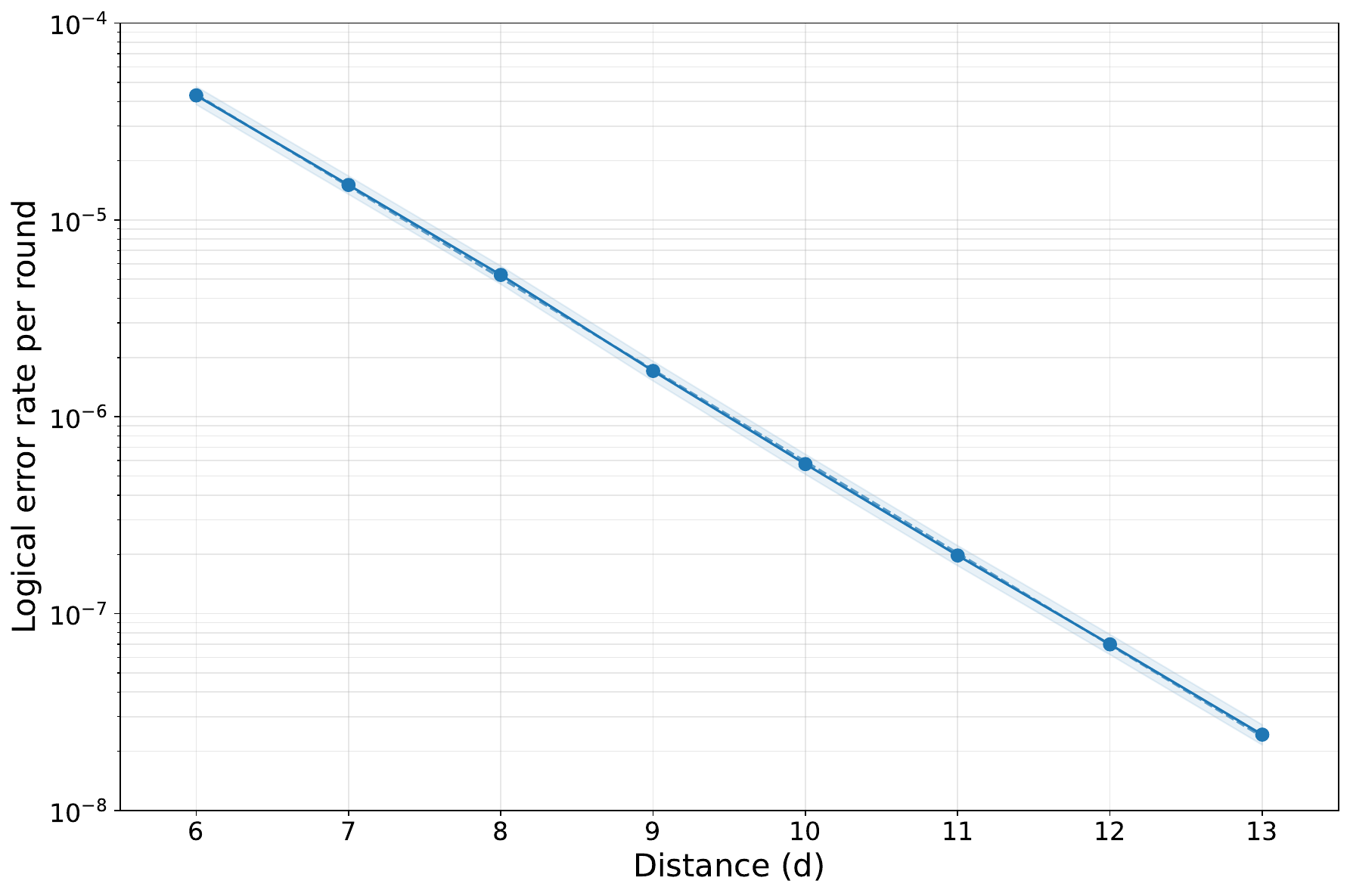}
  \caption{Left: logical error rates of the three-yoked surface code shown in Fig.~\ref{fig:the_multi_yoked_surface_code}. \texttt{t\_1} is a level-1 yoking interval, \texttt{t\_2} is a level-2 yoking interval, \texttt{r1} is a value of \texttt{t\_2/t\_1} and \texttt{sl} is a slope. Right: logical error rates of the surface code in the hex-grid implementation and decoded by the two-pass decoder. The dashed line is the fitting result. The logical error rates shown are per logical qubit.}
  \label{fig:surface_and_three_yoke_ler}
\end{figure}

\subsection{Density Evaluation}
We finally compare the number of physical qubits in each code, which is shown in Fig.~\ref{fig:num_phys_comparison}. 
In this evaluation, we use the fitted logical error rates $\log{p_\text{L}}= -0.4661d - 1.566$, $\log{p_\text{L}} = -0.794d +0.625$, $\log{p_\text{L}}= -1.5664d + 6.2812$, and $\log{p_\text{L}}= -2.8747d + 12.2164$ for the unyoked, one-yoked ($[[16,14,2]]$), two-yoked ($[[16,14,2]]\rightarrow[[112,104,2]]$), and three-yoked surface codes ($[[16,14,2]]\rightarrow[[112,104,2]]\rightarrow[[208,166,2]]$), respectively. 
The physical qubit count per logical qubit in the above evaluation includes the qubits used for workspaces, and we assumed the memory-space layout discussed in Sec.\,\ref{sec:overall} as shown in Fig.~\ref{fig:memory_comparison}. 

Before discussing this figure, we note two caveats about this evaluation. First, these memories serve as cold storage, and we do not include hallways for load and store operations between the memory and the computational space. If load and store functionality were required, the density would decrease slightly from the shown value. We nevertheless omit them to focus on the achievable density of storage. We later discuss candidate designs for data extraction in Sec.\,\ref{sec:discussion}. Second, we used the same layouts for two-yoked and three-yoked surface codes. As our patch layouts are optimized for the three-yoked surface code, the performance of the two-yoked surface code is likely underestimated. 

As can be seen from the figure, the three-yoked surface code significantly reduces the number of physical qubits. Suppose that a teraquop regime requires a logical error rate per inner round of less than about $10^{-15}$. Then, our three-yoked surface code is about 4 to 4.5$\times$ denser than the unyoked one in the teraquop regime.

\begin{figure}[t]
  \centering
  \includegraphics[width=0.49\textwidth]{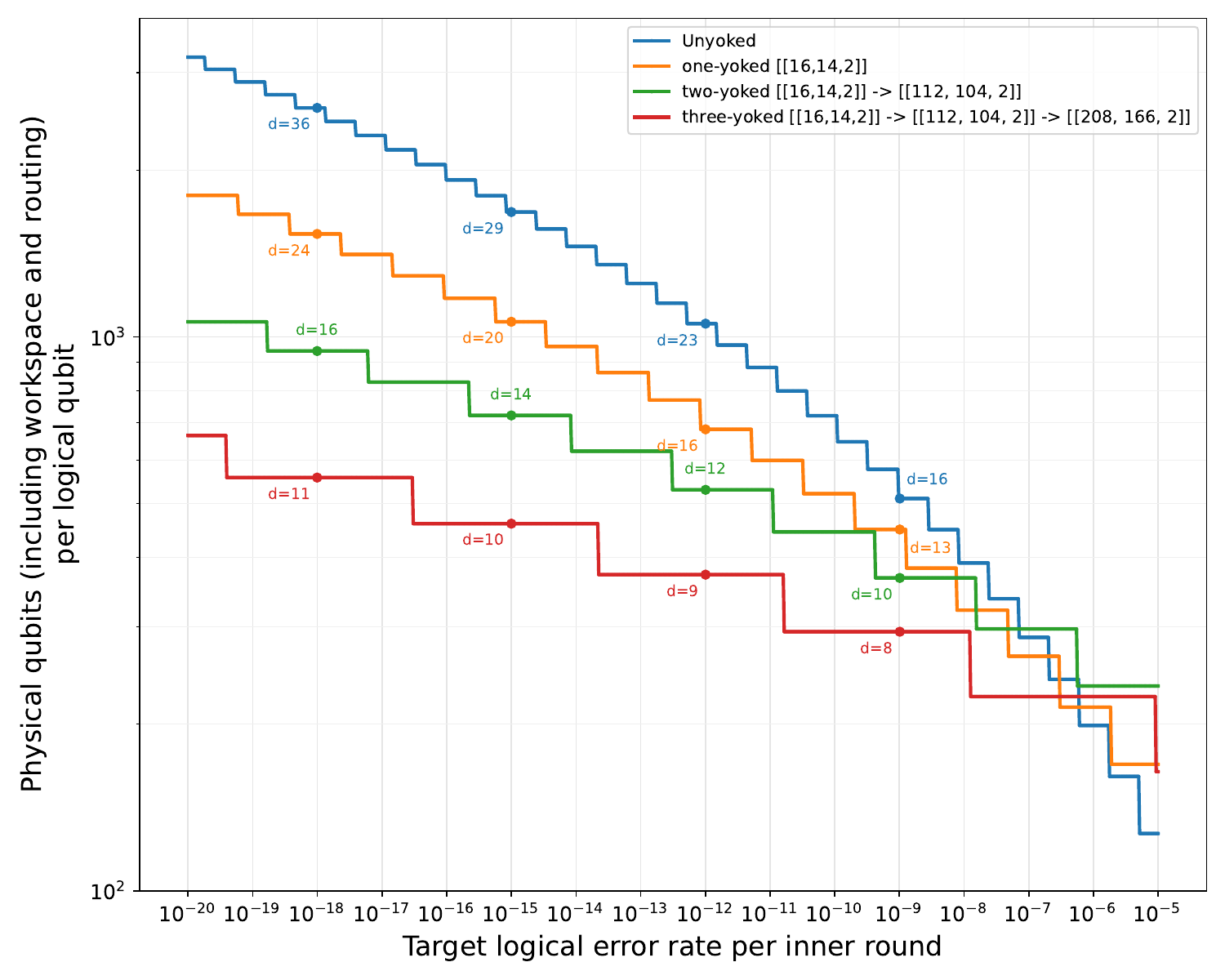}
  \includegraphics[width=0.49\textwidth]{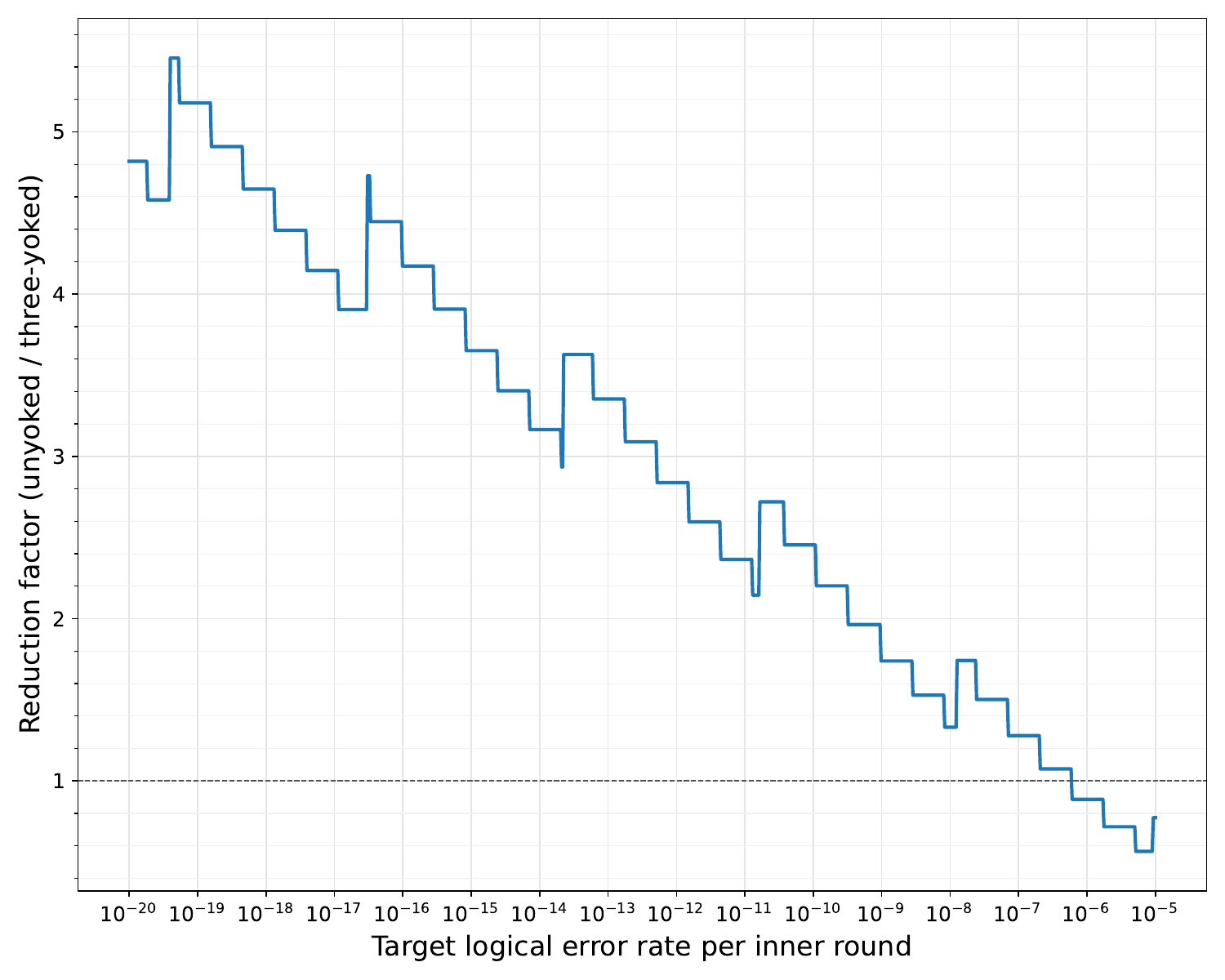}
  \caption{Left: the physical qubits comparison. Right: the reduction factor of the physical qubits of the three-yoked surface code from the surface code.}
  \label{fig:num_phys_comparison}
\end{figure}

\begin{figure}[t]
  \centering
  \includegraphics[width=0.90\textwidth]{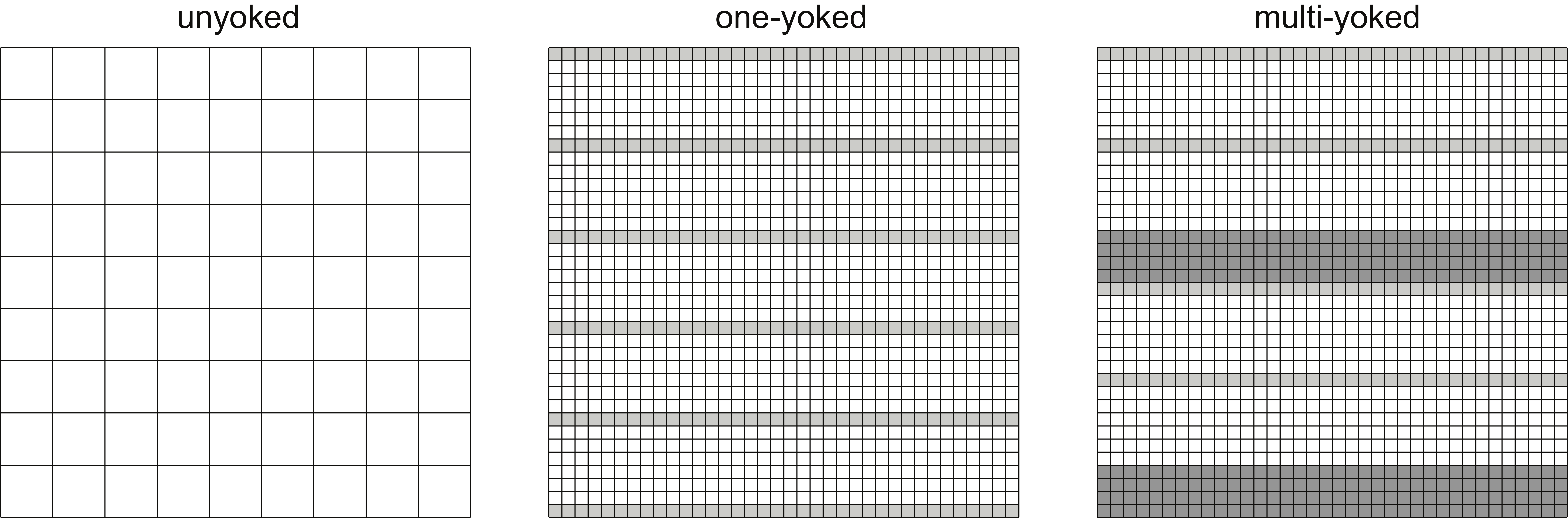}
  \caption{From left to right: surface-code layouts in the unyoked, one-yoked, and multi-yoked surface codes. The white patches hold the logical information. These are cold-storage memories that are not directly accessible, so all workspaces in the one-yoked and multi-yoked surface code memory are fully dedicated to yoking.}
  \label{fig:memory_comparison}
\end{figure}

\section{Discussion}
\label{sec:discussion}
In this work, we generalized yoked surface codes to a multi-level concatenated version, which we call multi-yoked surface codes. We constructed a new code family, called chain codes, to obtain the high-density parity-check codes that fit our multi-yoked construction. The chain codes have the favorable property that they can be concatenated over a code while managing correlated errors inside each block. Therefore, they enable a high-density encoding even with only a few concatenation levels. We also solved the long-interval problem, which arises when a single large-distance code is used as the outer code, by introducing a hierarchy among the yoking intervals at each level. By making the higher-level yoking intervals longer, we can keep lower-level yoking frequent. Under the SI1000 noise model, we confirmed that our three-yoked surface code shown in Fig.~\ref{fig:the_multi_yoked_surface_code} is about 4 to 4.5$\times$ denser than the unyoked surface code in the teraquop regime. This achieves a higher density than the original yoked surface code.

There are several future directions to further improve the idea of multi-yoked surface codes. We list our current open questions and future perspectives as follows.

\textbf{Can we further improve the density by higher-level concatenation?} 
This work considered up to three-yoked surface codes because simulating higher-level concatenation requires significant classical computational resources. With more computational resources, we might see that the density of multi-yoked surface codes improves further.

\textbf{Can we efficiently evaluate the performance of high-density and teraquop-regime memory?}
Simulating the entire memory is a major bottleneck in designing memory spaces for the teraquop regime, yet these designs should be completed before building the actual quantum memory. Exploring more sophisticated and accurate methods to evaluate logical-error-rate behavior in the teraquop regime is a vital topic.

\textbf{Can the outer decoder really keep up with the outer syndrome extraction?} 
A major drawback of using multi-yoked surface codes is that the outer code becomes non-matchable, so we cannot use fast MWPM decoders such as PyMatching. We expect this problem can be mitigated by employing two-stage decoders with carefully optimized parameters, but this goes beyond the scope of this paper and is left as future work.

\textbf{How should we design the timing of data load/store in accordance with both the yoking and the computation?}
This paper focuses on memory operations for the multi-yoked surface code and does not account for the workspace needed for load and store operations between the memory and hot storage or the computational space. Although data can be transferred to and from external code blocks via state teleportation using lattice surgery, the hallway design for this lattice surgery is not obvious, for the following reasons.
Supposing the required fault distance is $d_{\rm outer} d_\text{inner}$, the workspace for lattice surgery should have $d_{\rm outer} d_\text{inner}$ width and $d_{\rm outer} d_\text{inner}$ rounds of syndrome extraction. Therefore, unlike yoking operations, which tolerate time-like errors during lattice surgery, this operation cannot easily use a space-efficient implementation to reduce the workspace width. This problem can be resolved by separating a single lattice surgery with $d_{\rm outer} d_\text{inner}$ rounds into several lattice surgeries with fewer rounds, and by employing a space-efficient implementation for each lattice surgery. Errors in the operation are estimated using the soft information from the inner decoder. An example design for extracting data while preserving $d_{\rm outer}=4$ by splitting the extraction into two lattice surgeries is shown in Fig.~\ref{fig:data_extracting}. Optimizing the schedule of these separated lattice surgeries would be necessary to combine our ideas with computational space.

\textbf{How should such a memory be connected to the hotter memories without incurring significant latency?}
Compared to the original yoked surface code, our multi-yoked surface codes provide colder memory, i.e., higher-density memory with longer access latency. This does not conflict with yoked surface codes; instead, we envision using them complementarily to form a continuous memory hierarchy. Exploring efficient memory-hierarchy designs and latency-hiding techniques would be a promising direction, given progress in modern computer architectures.

\textbf{Is there an efficient procedure for constructing chain codes for a general level of concatenation?}
Several of our protocols, such as finding seed codes for given problematic errors and optimizing the weight of given codes with yoking-weight reduction, rely on heuristic methods. Because we focused on concatenations of small high-level codes, we found solutions that we believe are near-optimal. However, it remains open whether better search algorithms exist. This exploration might improve the construction and implementation of the multi-yoked surface codes.

\section*{Contributions}
YH proposed the original ideas, improved them, conducted the numerical calculations, and prepared the manuscript. YS supervised the project. All authors reviewed, revised, and approved the final manuscript. 

\section*{Code availability}
The source code used for the simulations in this work is available on GitHub at \url{https://github.com/yugahirai/multi_yoked_surface_codes}.

\section*{Use of Generative AI}
During preparation of this manuscript, the authors used Claude Fable 5.1 (Anthropic) to assist in language editing and developing heuristic optimization codes for the yoking weight reduction in Sec.\,\ref{sec:weight_reduction}. The authors independently verified the correctness of the numerical results in that section. All generated or revised text was reviewed and edited by the authors, who take full responsibility for the content of the manuscript.

\section*{Acknowledgments}
We thank Shota Ikari and Yosuke Ueno for their comments on this work.
This work is supported by MEXT Q-LEAP Grant No.~JPMXS0120319794 and JPMXS0118068682, MEXT Feasibility Study on the future HPCI, JST Moonshot R\&D Grant No.~JPMJMS2061, JST CREST Grant No.~JPMJCR23I4, JPMJCR24I4, and JPMJCR25I4.

\bibliography{ref}

\begin{figure}[h]
  \centering
  \includegraphics[width=0.95\textwidth]{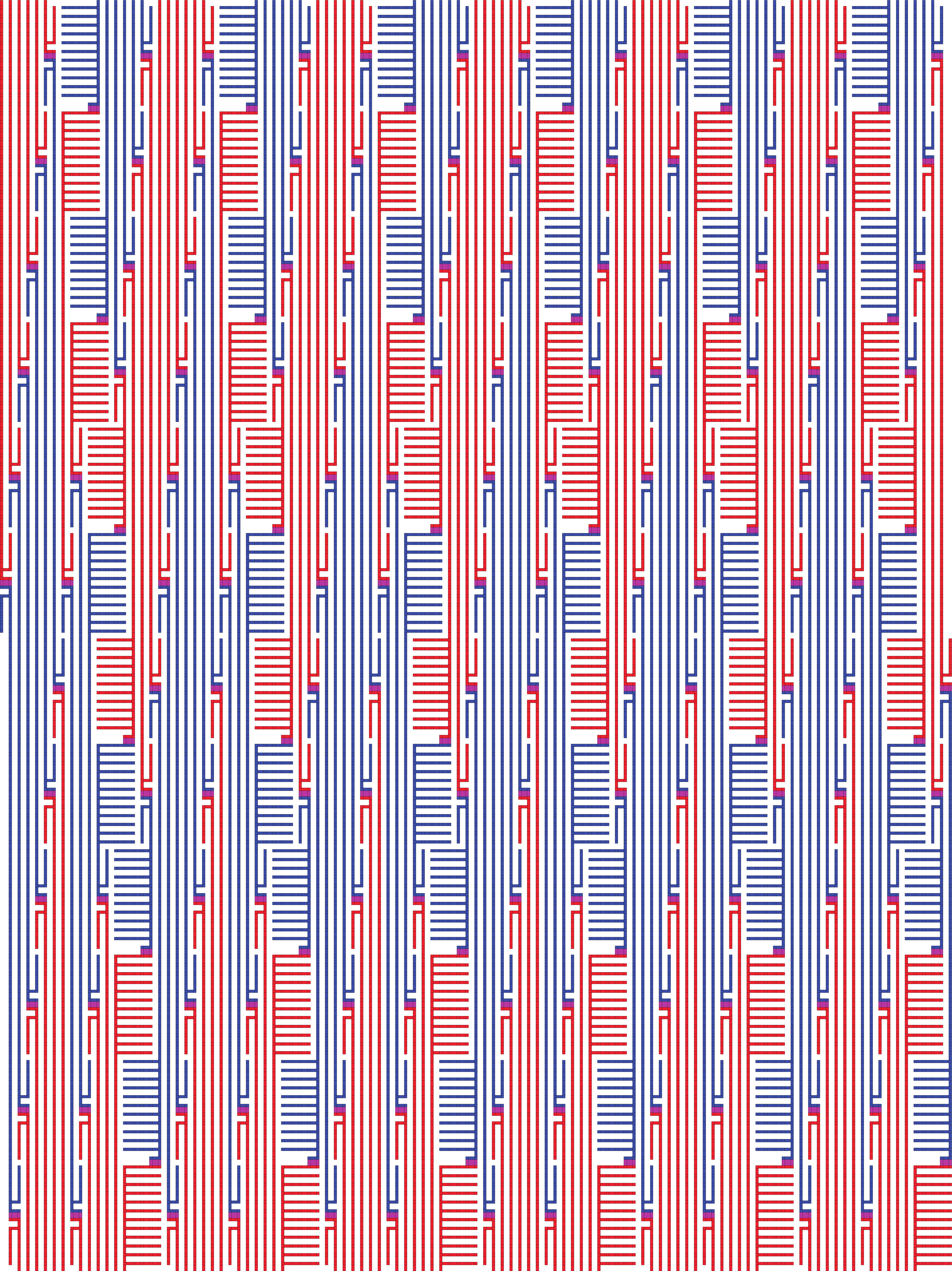}
  \caption{Full yoking schedule of the three-yoked surface code, shown as a defect diagram of the memory in Fig.~\ref{fig:patch_alloc_employed} viewed from the side. The narrow workspaces perform the level-1 checks, and the wide workspaces perform the level-2 and level-3 checks.}
  \label{fig:meanders}
\end{figure}

\begin{figure}[h]
  \centering
  \includegraphics[width=0.4\textwidth]{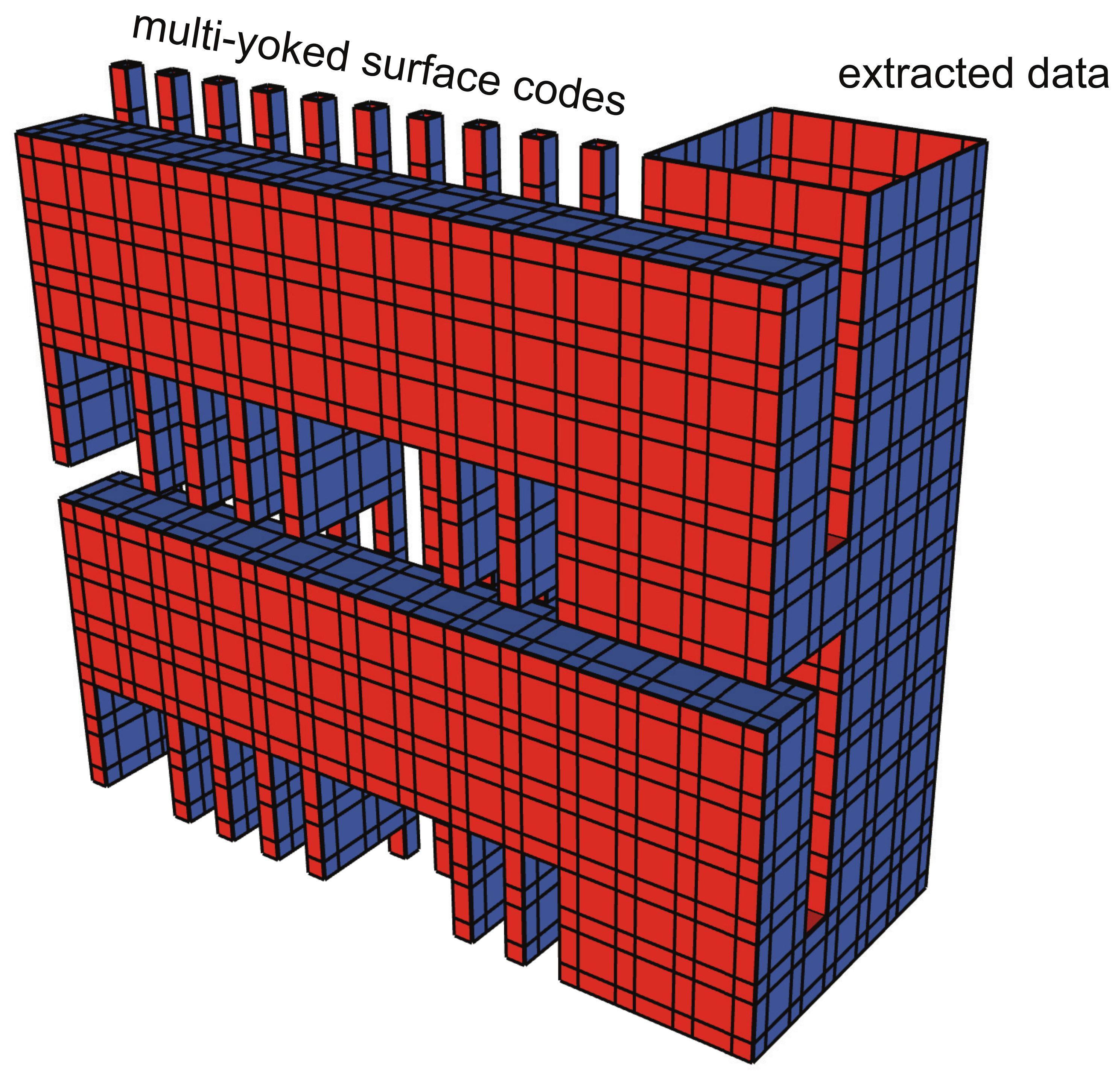}
  \caption{Data-extraction circuit with a width-$2d_\text{inner}$ workspace achieving fault distance $4d_\text{inner}$. The data in the memory are measured in the X basis after this step. Weight-$2d_\text{inner}$ errors can be detected by XORing the measured values of the two lattice surgeries between the multi-yoked and unyoked surface codes, and corrected using the soft information from the surface code decoder.}
  \label{fig:data_extracting}
\end{figure}

\end{document}